\documentclass[11pt, final,leqno,onefignum,onetabnum]{siamltex}
\usepackage{fullpage}
\usepackage{setspace}
\newcommand{\commentout}[1]{}

\usepackage{amsmath}
\usepackage{amsfonts}
\usepackage{mathcomp}
\usepackage{textcomp}
\usepackage{amssymb}
\usepackage{latexsym}
\usepackage{graphicx}
\usepackage{courier}
\usepackage[english]{babel}
\usepackage{mathbbol}
\usepackage{multirow}
\usepackage{latexsym}
\usepackage{epsfig}
\usepackage{subfigure}
\usepackage{dcolumn}
\usepackage{bm}
\usepackage{colordvi}
\usepackage{color}
\usepackage{booktabs}
\usepackage{stmaryrd}
\usepackage{cite}
\usepackage[linesnumbered,commentsnumbered,ruled]{algorithm2e}
\usepackage{epstopdf}
\usepackage{pdfsync}
\usepackage{showkeys}
\usepackage{hyperref}
\input xy
\xyoption{all}

\newtheorem{thm}{Theorem}[section]
\newtheorem{prop}[thm]{Proposition}

\newtheorem{rmk}[thm]{Remark}

\newcommand{\nwc}{\newcommand}
\nwc{\ben}{\begin{equation*}}
\nwc{\bea}{\begin{eqnarray}}
\nwc{\beq}{\begin{eqnarray}}
\nwc{\bean}{\begin{eqnarray*}}
\nwc{\beqn}{\begin{eqnarray*}}
\nwc{\beqast}{\begin{eqnarray*}}

\nwc{\eal}{\end{align}}
\nwc{\een}{\end{equation*}}
\nwc{\eea}{\end{eqnarray}}
\nwc{\eeq}{\end{eqnarray}}
\nwc{\eean}{\end{eqnarray*}}
\nwc{\eeqn}{\end{eqnarray*}}
\nwc{\eeqast}{\end{eqnarray*}}

\newcommand{\lt}{\left}

\newcommand{\rt}{\right}

\nwc{\bR}{\mb R}
\nwc{\bH}{{\mb H}}
\nwc{\bxp}{{{\mathbf x}}}
\nwc{\bap}{{{\mathbf y}}}

\nwc{\bPhi}{\mathbf{\Phi}}
\nwc{\bPsi}{\mathbf{\Psi}}
\nwc{\bh}{\mathbf h}

\nwc{\bI}{\mathbf I}
\nwc{\bP}{\mathbf P}
\nwc{\bs}{\mathbf s}
\nwc{\bd}{\mathbf{d}}
\nwc{\bX}{\mathbf X}

\nwc{\om}{\omega}

\nwc{\nwt}{\newtheorem}
\nwc{\xp}{{x^{\perp}}}
\nwc{\yp}{{y^{\perp}}}

\nwc{\ba}{{\mb a}}
\nwc{\bal}{\begin{align}}

\nwc{\vep}{\varepsilon}
\nwc{\ep}{\epsilon}
\nwc{\ept}{\epsilon}
\nwc{\vrho}{\varrho}
\nwc{\orho}{\bar\varrho}
\nwc{\ou}{\bar u}
\nwc{\vpsi}{\varpsi}
\nwc{\lamb}{\lambda}
\nwc{\Var}{{\rm Var}}

\nwc{\nn}{\nonumber}
\nwc{\mf}{\mathbf}
\nwc{\mb}{\mathbf}
\nwc{\ml}{\mathcal}

\nwc{\IA}{\mathbb{A}} 
\nwc{\bi}{\mathbf i}
\nwc{\bo}{\mathbf o}
\nwc{\IB}{\mathbb{B}}
\nwc{\IC}{\mathbb{C}} 
\nwc{\ID}{\mathbb{D}} 
\nwc{\IM}{\mathbb{M}} 
\nwc{\IP}{\mathbb{P}} 
\nwc{\II}{\mathbb{I}} 
\nwc{\IE}{\mathbb{E}} 
\nwc{\IF}{\mathbb{F}} 
\nwc{\IG}{\mathbb{G}} 
\nwc{\IN}{\mathbb{N}} 
\nwc{\IQ}{\mathbb{Q}} 
\nwc{\IR}{\mathbb{R}} 
\nwc{\IT}{\mathbb{T}} 
\nwc{\IZ}{\mathbb{Z}} 

\nwc{\cE}{{\ml E}}
\nwc{\cP}{{\ml P}}
\nwc{\cQ}{{\ml Q}}
\nwc{\cL}{{\ml L}}
\nwc{\cX}{{\ml X}}
\nwc{\cW}{{\ml W}}
\nwc{\cZ}{{\ml Z}}
\nwc{\cR}{{\ml R}}
\nwc{\cV}{{\ml V}}
\nwc{\cT}{{\ml T}}
\nwc{\crV}{{\ml L}_{(\delta,\rho)}}
\nwc{\cC}{{\ml C}}
\nwc{\cO}{{\ml O}}
\nwc{\cA}{{\ml A}}
\nwc{\cK}{{\ml K}}
\nwc{\cB}{{\ml B}}
\nwc{\cD}{{\ml D}}
\nwc{\cF}{{\ml F}}
\nwc{\cS}{{\ml S}}
\nwc{\cM}{{\ml M}}
\nwc{\cG}{{\ml G}}
\nwc{\cH}{{\ml H}}
\nwc{\bk}{{\mb k}}
\nwc{\bn}{{\mb n}}
\nwc{\cbz}{\overline{\cB}_z}
\nwc{\fR}{\Re}
\nwc{\bY}{\mathbf Y}

\nwc{\pft}{\cF^{-1}_2}
\nwc{\bU}{{\mb U}}
\nwc{\bG}{{\mb G}}
\nwc{\bg}{\mathbf{g}}
\nwc{\mbf}{\mathbf{f}}
\nwc{\mbe}{\mathbf{e}}
\nwc{\be}{\mathbf{e}}
\nwc{\Om}{\Omega}
\nwc{\ind}{\operatorname{I}}
\nwc{\mbx}{\mathbf{f}}
\nwc{\bb}{\mathbf{g}}
\nwc{\xmax}{f_{\rm max}}
\nwc{\xmin}{f_{\rm min}}
\nwc{\suppx}{\hbox{\rm supp} (\mbf)}
\nwc{\by}{\mathbf{h}}
\nwc{\bZ}{\mathbf{Z}}
\nwc{\bF}{\mathbf{F}}
\nwc{\bE}{\mathbf{E}}
\nwc{\bV}{\mathbf{V}}
\nwc{\cI}{\IZ^2_N}
\nwc{\chis}{{\chi^{\rm s}}}
\nwc{\chii}{{\chi^{\rm i}}}
\nwc{\pdfi}{{f^{\rm i}}}
\nwc{\pdfs}{{f^{\rm s}}}
\nwc{\pdfii}{{f_1^{\rm i}}}
\nwc{\pdfsi}{{f_1^{\rm s}}}
\nwc{\thetatil}{{\tilde\theta}}
\nwc{\red}{\color{red}}
\nwc{\blue}{\color{blue}}
\nwc{\prox}{\hbox{prox}}
\nwc{\sloc}{J_{\rm f}}
\nwc{\bu}{\xi}
\nwc{\bv}{\eta}
\nwc{\cU}{\mathcal{U}}
\nwc{\cN}{\mathbf{N}}
\nwc{\bN}{\mathbf{N}}
\nwc{\mbm}{\mathbf{m}}
\nwc{\bw}{\mathbf{w}}
\nwc{\im}{i}
\nwc{\bom}{\mathbf{w}}
\nwc{\bt}{\mathbf{t}}
\nwc{\z}{y}
\nwc{\cY}{\mathcal{Y}}
\nwc{\bM}{\mathbf{M}}
\nwc{\half}{{1\over 2}}
\nwc{\xnul}{x_{\rm null}}

\begin{document}

\title{
Phase Retrieval with One or Two Diffraction Patterns by Alternating Projection with the Null Initialization}

\date{Oct 24, 2015}
\author{Pengwen Chen
\thanks{
Department of Applied Mathematics, National Chung Hsing University, Taichung  402, Taiwan. Research is supported in part by the grant 103-2115-M-005-006-MY2 from Ministry of Science and Technology, Taiwan,  and US NIH grant U01-HL-114494}
\and Albert Fannjiang
\thanks{Corresponding author. Department of Mathematics, University of California, Davis, CA 95616, USA. Research is supported in part by  US National Science Foundation  grant DMS-1413373 and Simons Foundation grant 275037.}
\and Gi-Ren Liu
\thanks{Department of Mathematics, University of California, Davis, CA 95616, USA}
}

\maketitle

\begin{abstract}  
Alternating projection (AP) of various forms, including the Parallel AP (PAP), Real-constrained AP (RAP)
and the Serial AP (SAP), are proposed to solve phase retrieval with
at most two coded diffraction patterns. The proofs of geometric convergence are given
with sharp bounds on the rates of convergence in terms of a spectral gap condition.
 
To compensate for the local nature of convergence, the null initialization is proposed
for initial guess and proved to produce asymptotically accurate initialization for the case of Gaussian random measurement.  Numerical experiments  show that
the null initialization produces more accurate initial guess than the spectral initialization
and that AP converges faster to the true object than other iterative schemes for non-convex optimization such as the Wirtinger Flow.  In numerical experiments AP with the null initialization  converges globally to the true object. \end{abstract}

\pagestyle{myheadings}


\section{Introduction}

With  wide-ranging applications in science and technology, phase retrieval has recently attracted a flurry of activities in
the mathematics community (see a recent review \cite{Sh} and references therein).
 Chief among these applications is
the coherent X-ray diffractive imaging of a single particle  using a coherent, high-intensity source such as synchrotrons and free-electron lasers. 

In the so-called {\em diffract-before-destruct} approach, the structural information of the sample particle  is captured by an ultra-short and ultra-bright X-ray pulse and recorded 
by a CCD camera   \cite{Chapman14, Chapman11, Hajdu}. 
To this end,  reducing the radiation exposure and damage is crucial. Due to the high frequency of the illumination field, the recorded data are the intensity of the diffracted field
whose phase needs to be recovered by mathematical and algorithmic techniques.
This gives rise to the problem of phase retrieval with non-crystalline structures. 

The earliest algorithm of phase retrieval for a non-periodic object (such as a single molecule)  is the Gerchberg-Saxton algorithm~\cite{GS72} and its variant, Error Reduction \cite{Fie82}. The basic idea is Alternating Projection (AP), going back all the way to the works of 
von Neuman, Kaczmarz and Cimmino  in the 1930s \cite{Cimmino, Kac, Neuman}. 
And these further trace the history  back to  Schwarz \cite{Schwarz}  who in 1870 used AP to solve the Dirichlet
problem on a region given as a union of regions each having a simple to solve Dirichlet problem.

For any vector $y$ let $|y|$ be the vector such that $|y|(j)=|y(j)|,\forall j$. In a nutshell, phase retrieval is to solve the equation of the form  $b=|A^* x_0|$ where
$x_0\in \cX\subseteq \IC^n$ represents the unknown object, $A^*\in \IC^{N\times n}$ the diffraction/propagation process  and $b^2\in \IR^N$ the diffraction pattern(s).  The subset $\cX$ represents all prior constraints on the object. Also, the number of data $N$ is typically greater than the number $n$ of components in $x_0$. 

Phase retrieval  can be formulated   as the following feasibility problem 
 \beq
 \label{feas}
\hbox{Find}\quad  \hat y\in  A^*\cX \cap \cY,\quad \cY:=  \{y\in \IC^N: |y|=b\}.
 \eeq
From $\hat y$  the object is estimated via pseudo-inverse 
 \beq
 \label{feas2}
 \hat x=(A^*)^\dagger \hat y. 
 \eeq
 Let $P_1$ be the projection onto $A^*\cX$ and $P_2$ the projection onto $\cY$ defined as 
 \[
 P_2 z=b \odot {z\over |z|},\quad z\in \IC^N
 \]
 where $\odot$ denotes the Hadamard product and $z/|z| $ the componentwise division.
  Where $z$ vanishes, $z/|z|$ is chosen to be 1 by convention.  
 Then  AP is simply the iteration of the composite map  
 \beq\label{fap}
 P_1P_2 y
 \eeq
  starting with an initial guess
 $y^{(1)}=A^* x^{(1)}, x^{(1)}\in \cX$. 
  
 The main structural  difference between AP in the classical setting \cite{Cimmino, Kac, Neuman} 
 and the current setting is the {\em non-convexity} of the set $\cY$, rendering the latter much more difficult to analyze. Moreover, AP for phase retrieval is well known to have stagnation problems in practice,
 resulting in poor reconstruction  \cite{Fie82,Fie13, Mar07}.  
 
In our view, numerical stagnation has   more to do with the measurement scheme than non-convexity: the existence of multiple solutions when only one (uncoded) diffraction pattern is measured even if additional  positivity
 constraint is imposed on the object. However, if the diffraction pattern is measured
 with a random mask (a coded diffraction pattern), then the uniqueness of solution under the real-valuedness constraint is restored with probability one \cite{unique}. In addition, if two independently coded diffraction patterns are measured, then the uniqueness of solution, up to a global phase factor, 
 holds almost surely without any additional prior constraint \cite{unique} (see Proposition \ref{prop:unique}).  
 
 The main goal of the present work is to show by analysis and numerics that under the uniqueness framework for phase retrieval with coded diffraction patterns of \cite{unique}, AP has a significantly sized basin of attraction at $x_0$ and that
 this basin of attraction can be reached by an effective initialization scheme, called the null initialization. In practice, numerical stagnation disappears under the uniqueness measurement schemes of \cite{unique}. 
 
 Specifically, 
our goal is two-fold: i) prove the local convergence of various versions of AP under the uniqueness framework of \cite{unique}  (Theorems \ref{thm1},  \ref{thm2} and \ref{thm3}) and  ii) propose a novel method of initialization, {the null initialization},  that compensates for the local nature of convergence and results in 
global convergence in practice. In addition, we prove that for Gaussian random measurements  the null initialization {\em alone}
produces an initialization of arbitrary accuracy as the sample size increases  (Theorem \ref{Gaussian}). 
In practice AP with the null initialization converges  globally to
the true object.  

 
 \subsection{Set-up} \label{sec:not}
Let us recall the measurement schemes of \cite{unique}.

 Let $x_0(\bn)$ be a discrete  object function with $\bn = (n_1,n_2,\cdots,n_d) \in \IZ^d$. 
Consider  the {object space} consisting  of all functions  supported in 
\[
\cM = \{ 0\le m_1\le M_1, 0\le m_2\le M_2,\cdots, 0\leq m_d\leq M_d\}. 
\]
We assume $d\geq 2$. 

Only the {\em intensities} of the Fourier transform, called the diffraction pattern,
are measured  
 \beq
   \sum_{\bn =-\bM}^{\bM}\sum_{\mbm\in \cM} x_0(\mbm+\bn)\overline{x_0(\mbm)}
   e^{-\im 2\pi \bn\cdot \bom},\quad \bom=(w_1,\cdots,w_d)\in [0,1]^d,\quad \bM = (M_1,\cdots,M_d)\nn
   \eeq
   which is the Fourier transform of the autocorrelation
   \beqn
	  R(\bn)=\sum_{\mbm\in \cM} x_0(\mbm+\bn)\overline{x_0(\mbm)}.
	  \eeqn
Here and below the over-line  means
complex conjugacy. 

Note that
$R$ is defined on the enlarged  grid
 \begin{eqnarray*}
 \widetilde \cM = \{ (m_1,\cdots, m_d)\in \IZ^d: -M_1 \le m_1 \le M_1,\cdots, -M_d\le m_d\leq M_d \} 
 \end{eqnarray*}
whose cardinality is roughly $2^d$ times that of $\cM$.
Hence by sampling  the diffraction pattern
 on the grid 
\beqn
\cL = \Big\{(w_1,\cdots,w_d)\ | \ w_j = 0,\frac{1}{2 M_j + 1},\frac{2}{2M_j + 1},\cdots,\frac{2M_j}{2M_j + 1}\Big\}
\eeqn
we can recover the autocorrelation function by the inverse Fourier transform. This is the {\em standard oversampling} with which  the diffraction pattern and the autocorrelation function become equivalent via the Fourier transform \cite{Miao00,MSC}.


A coded diffraction pattern is measured with a mask
whose effect is multiplicative and results in  
a {\em masked object}  of the form $
\tilde x_0(\bn) =x_0(\bn) \mu(\bn)$ 
where $\{\mu(\bn)\}$ is an array of random variables representing the mask.   
In other words, a coded diffraction pattern is just the plain diffraction pattern of
a masked object. 

We will focus on the effect of {\em random phases} $\phi(\bn)$ in the mask function 
$
\mu(\bn)=|\mu|(\bn)e^{\im \phi(\bn)}
$
where  $\phi(\bn)$ are independent, continuous real-valued random variables and $|\mu|(\bn)\neq 0,\forall \bn\in \cM$ (i.e. the mask is transparent).

For simplicity we assume $|\mu|(\bn)=1,\forall\bn$ which gives rise to
a {\em phase} mask and an {\em isometric}  propagation matrix 
\beq
\label{one}
\hbox{\rm (1-mask )}\quad A^*= c\Phi\,\, \diag\{\mu\},
\eeq
i.e. $AA^*=I$ (with a proper choice of the normalizing constant  $c$), where $\Phi$ is the {\em oversampled}  $d$-dimensional discrete Fourier transform (DFT). Specifically  $\Phi \in \IC^{|\tilde \cM|\times |\cM|}$ is the sub-column matrix of
the standard DFT on  the extended grid $\tilde \cM$ where $|\cM|$ is
the cardinality of $\cM$.  

If the non-vanishing mask $\mu$ does not have a uniform transparency, i.e. $|\mu|(\bn)\neq 1, \forall \bn,$ then we can define  a new object vector $|\mu|\odot x_0$ and a new
isometric propagation matrix
\[
A^*= c\Phi\,\, \diag\lt\{{\mu\over |\mu|}\rt\}
\]
with which to recover the new object first.

When two phase masks $\mu_1, \mu_2$ are deployed, 
the propagation matrix $A^*$ is the stacked coded DFTs, i.e.   
\beq \label{two}\hbox{(2-mask case)}\quad 
A^*=c \lt[\begin{matrix}
\Phi\,\, \diag\{\mu_1\}\\
\Phi\,\, \diag\{\mu_2\}
\end{matrix}\rt]. 
\eeq
 With proper normalization, $A^*$ is
isometric.

We convert the $d$-dimensional ($d\geq 2$)  grid into an ordered set of index. Let $ n=|\cM|$ and  $N$ the total number of measured data. In other words, $A\in \IC^{N\times n}$. 

Let $\cX$ be a nonempty closed convex  set in $\IC^n$ and let 
\begin{equation}
[x]_\cX=\hbox{\rm arg}\min_{x'\in \cX} \|x'-x\|
\end{equation}
denote the projection onto $\cX$. \\
Phase retrieval is to find a solution $x$ to 
the equation 
\beq
\label{0}
 b=|A^* x|,\quad x\in \cX.
\eeq
We focus on the following two cases. \\ 

{{\bf 1) One-pattern case:} $A^*$ is given by \eqref{one}, $\cX=\IR^n$ or $\IR^n_{\tiny +}$. \\

{{\bf 2) Two-pattern case:} $A^*$ is given by \eqref{two},  $\cX=\IC^n$ (i.e. $[x]_\cX=x$).  \\

For the two-pattern case, AP for the formulation \eqref{feas} shall be called  the {Parallel} AP (PAP)  as the rows of $A^*$ and the diffraction data
 are treated equally and simultaneously, in contrast to the {Serial} AP (SAP)  which splits
 the diffraction data into two blocks according to the masks and treated alternately.

The main property of the true object is the rank-$k$ property: $x_0$ is rank-$k$ if the convex hull of  $\supp\{x_0\}$ in $\IC^n$ is $k$-dimensional. 

Now we recall the uniqueness theorem of phase retrieval with coded diffraction patterns.

\begin{proposition} \label{prop:unique} \cite{unique} (Uniqueness of Fourier phase retrieval) 
 Let $x_0$ be a rank$-k,~k\geq 2$, object and  $x$ a solution  of the  phase retrieval problem \eqref{0} for either the one-pattern or two-pattern case. 
 Then   $x=e^{i\theta} x_0$ for some constant $\theta\in \IR$ with  probability one. 
\end{proposition}
\begin{rmk}
The main improvement over  the classical uniqueness theorem  \cite{Hayes} is that while the classical result works with generic (thus random) objects Proposition \ref{prop:unique}  deals with a given  deterministic object. By definition, deterministic objects belong to the measure zero set excluded in the classical setting of \cite{Hayes}. It is crucial to endow  the probability measure on the ensemble of random masks, which we can manipulate, instead of the space of unknown objects, which we can not control.\end{rmk}

The proof of Proposition \ref{prop:unique} is given in \cite{unique} where
more general uniqueness theorems can be found, including the $1\half$-mask case. 
Phase retrieval solution is unique only up to a constant of modulus one
 no matter how many coded diffraction patterns are measured.
 Thus a reasonable  error metric for an estimate $\hat x$ of the true solution $x_0$  is given by
  \begin{equation}
\min_{\theta\in \IR}\|e^{i\theta} \hat x - x_0 \|. 
 \end{equation}

Our framework and methods can be extended to more general, non-isometric measurement matrix $A^*$ as follows. Let $A^*=QR$ be the QR-decomposition of $A^*$ where
$Q$ is isometric and $R$ is upper-triangular. We have \begin{equation}
Q^*= A^*(AA^* )^{-1/2}
\end{equation}
if $A$ (and hence $R$) is full-rank. 
Now  we can define a new object vector $Rx$ and a new isometric measurement
matrix $Q$ with which to recover $Rx$ first. 

\commentout{
Throughout the paper, we assume 
the  canonical  embedding 
\[
\IC^n\subseteq \IC^N,\quad n\leq N.
\]
For example,  if $x\in \IC^n$, then the embedded vector in $\IC^N$, still denoted by $x$,
has zero components $x(j)=0$ for $j\geq n+1$. This is  referred to as {\em zero padding} and
$\tilde n/n$ is the {\em padding ratio}.
Conversely, if $x\in \IC^{\tilde n}$ or $ \IC^N$, then $[x]_n \in \IC^n$ denotes 
 the projected vector onto $\IC^n$. 
 }


\subsection{Other  literature}\label{sec:lit}
Much of recent mathematical literature on phase retrieval focuses on generic frames and random measurements, see e.g. \cite{Balan2, Balan1, BY, BM2, BM1, truncatedWF, phaselift1, Conca, DH12, Eldar, Gross, LV,  NJS, Oh, Sh, Mallat, Xin, Sh}. Among the mathematical works on Fourier phase retrieval e.g. \cite{BCL02, phaselift0, CLS2, CLS1, DR-phasing, Dobson, unique,rpi, pum, Hayes, Hesse, Klib1, Klib2, Mar07, Mig11,  Noll, Papa, Vetterli, ADM}, only a few focus on
analysis and development of efficient algorithms. 

\commentout{
For the optical spectrum,  experiments with coded diffraction patterns  are not new
and can be  implemented 
 by computer generated holograms \cite{BWW}, random phase plates \cite{AH1} and
liquid crystal phase-only panels \cite{FAK}.  Recently, a phase mask with
randomly distributed pinholes  has been
implemented  for soft X-ray \cite{ptycho-rpi}. 
}

There is also vast literature on AP. We only mention the most relevant literature and refer the reader to the reviews \cite{BB,Deutsch} for a more complete list of references. Von Neumann's convergence theorem \cite{Neuman} for 
AP with two closed subspaces is extended to the setting of closed convex sets in
\cite{CG, Bregman} and,  starting with \cite{GS72}, the application of AP to the non-convex setting of phase retrieval has been extensively
studied \cite{Fie82, Fie13, BCL02,BCL04,Mar07}.

 In \cite{Luke09} in particular, 
local convergence theorems were developed for AP for  non-convex problems. However, the technical challenge in applying the theory in \cite{Luke09} to phase retrieval lies  precisely in verifying the main assumption of linear regular intersection therein. 

In contrast,  in the present work, what guarantees the geometric convergence and gives an often sharp bound on the convergence rate  is
the spectral gap condition which can be readily verified  under the uniqueness framework of
\cite{unique} (see Propositions \ref{cor5.2} and  \ref{prop4.8} below). 

As pointed out above, there
are more than one way of formulating phase retrieval, especially with  two (or more) diffraction patterns,
as a feasibility problem.  While PAP is analogous to Cimmino's approach to  AP \cite{Cimmino}, SAP is closer in spirit to Kaczmarz's \cite{Kac}. Surprisingly, 
SAP performs significantly better than PAP in our simulations (Section \ref{sec:num}).
 In Sections \ref{sec:loc} and \ref{sec:SAP} we prove that 
both schemes are locally convergent to the true solution with bounds on rates of convergence.
measurement local convergence for PAP  was proved  in \cite{NJS}.  

Despite the theoretical appeal of  a convex minimization approach to phase retrieval  \cite{phaselift0, phaselift1, CLS1, Papa}, 
the tremendous increase in dimension results in  impractically slow computation. 
Recently, new non-convex approaches become popular again  because of their computational  efficiency 
among other benefits \cite{CLS2,Mig11, NJS}. 

One purpose of the present work is to compare these newer approaches with AP, arguably  the simplest of all non-convex approaches. An important difference of
the measurement schemes in 
these papers from ours  is
that their coded diffraction patterns are {\em not} oversampled. 
In this connection, we emphasize that
reducing the number of coded diffraction patterns
is crucial for the diffract-before-destruct approach
and it is better to oversample than to increase the number of coded diffraction patterns. Another difference is that these newer iterative schemes such as the Wirtinger Flow (WF) \cite{CLS2} are not of the projective type. In Section \ref{sec:num}, we provide a detailed numerical comparison between AP of various forms and WF.

Recently  we  proved  local convergence of
the Douglas-Rachford (DR) algorithm for coded-aperture phase retrieval \cite{DR-phasing}.
The present work extends the method of \cite{DR-phasing} to AP. In addition to
convergence analysis of AP, we also characterize the limit points and the fixed points
of AP in the present work. 

More important,  to compensate for the local nature of convergence we develop a novel procedure, the null initialization,  for finding a sufficiently close initial guess. We prove that the null initialization with the Gaussian random measurement matrix asymptotically approaches the true object (Section \ref{sec:null}). 
The analogous result for coded diffraction patterns remains open. The null initialization is  significantly different from the spectral initialization  proposed in \cite{NJS, CLS2, truncatedWF}. 
In Section \ref{sec:Gaussian} we give a theoretical comparison and in Section \ref{sec:num} a numerical comparison between these initialization methods. We will see that the initialization with the null initialization is more accurate than with the spectral
initialization and SAP with the null initialization converges faster than the Fourier-domain Douglas-Rachford algorithm proposed in \cite{DR-phasing}.

{ During the review process, the two references \cite{Noll,Hesse} were brought to our attention by the referees. 

 Theorem 3.10 of \cite{Hesse} 
asserts 
global convergence to {\em some} critical point of a proximal-regularized alternating minimization formulation of \eqref{feas} 
provided that the iterates are 
{\em bounded} (among other assumptions). However,  neither (global or local) convergence to the {\em true} solution nor the geometric  sense of convergence is established in \cite{Hesse}.
In contrast, we prove that the AP iterates are always bounded, their accumulation points must be fixed points (Proposition \ref{Cauchy}) and the true solution is a stable fixed point. Moreover, any fixed
point  that shares the
same 2-norm with the true object is the true object itself (Proposition \ref{prop2.21}).

On the other hand, Corollary 12 of \cite{Noll} asserts the existence of a local basin of attraction of the feasible set \eqref{feas} which includes  AP in the one-pattern case and PAP in the two-pattern case (but not SAP). From this and  the uniqueness theorem (Proposition \ref{prop:unique}) convergence to the true solution, up to a global phase factor, follows (i.e. a singleton with an arbitrary global phase factor). However, Corollary 12 of \cite{Noll} asserts a {\em sublinear} power-law convergence with an unspecified power.  In contrast, we prove a linear convergence and give a spectral gap bound on the convergence rate for AP, including SAP which is emphatically {not} covered by  \cite{Noll} and arguably the best performer among the tested algorithms. }

The paper proceeds as follows. In Section \ref{sec:null}, we discuss the null initialization and prove global convergence to the true object of the null initialization for the complex Gaussian random measurement. In Section \ref{sec:pap}, we formulate  AP
of various forms and in Section \ref{sec:fixed} we discuss the limit points and the fixed points of AP.
We prove local convergence to the true solution for the Parallel AP in  Section \ref{sec:loc} and  for the real-constraint AP in Section \ref{sec:one-pattern}. In Section \ref{sec:SAP} we prove local convergence for
the Serial AP. 
 In Section \ref{sec:num},  we present 
numerical experiments and compare our approach with
the Wirtinger Flow and its truncated version  \cite{CLS2,truncatedWF}. 

\section{The null initialization}\label{sec:null}

For a nonconvex minimization problem such as phase retrieval,
the accuracy of  the initialization as the estimate of the object has a great impact on  the  performance  of any iterative schemes.

 The following observation motivates  our approach to
 effective initialization. 
Let $I$ be a subset of $\{1,\cdots,N\}$  and $I_c$ its complement such that $b(i)\leq b(j)$ for all $i\in I, j\in I_c$.  In other words, $\{b(i): i\in I\}$ are the ``weaker" signals and
$\{b(j): j\in I_c\}$ the ``stronger" signals. 
 Let $|I|$ be the cardinality of the set $I$.
 Then  $\{ a_i\}_{i\in I}$ is  a set of sensing vectors  nearly orthogonal  to $x_0$ if $|I|/N$ is sufficiently small (see Remark \ref{rmk5.2}).
This suggests the following  constrained  least squares solution   \[
x_{\rm null}:=\hbox{\rm arg}\min\lt\{\sum_{i\in I} \|a_i^* x\|^2: x\in \cX, {\|x\|=\|x_0\|}\rt\}
 \]
 may be  a reasonable initialization. Note that $x_{\rm null}$ is not uniquely defined 
 as $\alpha x_{\rm null}$, with $ |\alpha|=1,$ is also a null vector. Hence we should
 consider the global phase adjustment for a given null vector $x_{\rm null}$
 \beq
\nn
 \min_{\alpha \in \IC,\; |\alpha|=1 }\| \alpha x_{\rm null}-x_0\|^2=2\|x_0\|^2-2\max_{|\alpha|=1} \Re(x_0^*\alpha\xnul).
 \eeq
In what follows, we assume $\xnul$ to be optimally adjusted so  that 
\beq
 \label{52'}\|x_{\rm null}-x_0\|^2=2\|x_0\|^2-2 |x_0^*\xnul| 
\eeq

{ We pause to emphasize that the constraint $\|x_{\rm null}\|=\|x_0\|$ is introduced in order to simplify the error bound below (Theorem \ref{Gaussian}) and is completely 
irrelevant to initialization since the AP map $\cF$ (see \eqref{papf} below for definition) is scaling-invariant in the sense that   $
\cF(cx)=\cF(x)$,
for any $c>0$. 
Also, in many imaging problems, the norm of the true object, like the constant phase factor, is either  recoverable by other prior information or irrelevant to the quality of reconstruction.}

 Denote  the sub-column  matrices  consisting of $\{a_i\}_{i\in I} $ and $\{a_j\}_{j\in I_c}$  by $A_I$ and $A_{I_c}$, respectively, and, by reordering  the row index, write  $A=[A_I, A_{I_c}]\in \IC^{n\times N}.$ 
 
 Define the dual vector 
  \beq
x_{\rm dual}:= \hbox{\rm arg}\max\lt\{ \|A_{I_c}^*  x\|^2: x\in \cX, {\|x\|=\|x_0\|}\rt\} \eeq
 whose phase factor is optimally adjusted as $x_{\rm null}$.  
 \subsection{Isometric $A^*$}
 
For isometric  $A^*$, 
 \beq
 \label{2.3}
x_{\rm null}:=\hbox{\rm arg}\min\lt\{\sum_{i\in I} \|a_i^* x\|^2: x\in \cX, {\|x\|=\|b\|}\rt\}.
 \eeq
We have  
 \[
 \|A_I^*  x\|^2+\|A_{I_c}^*  x\|^2=\|x\|^2
\]
and hence 
\beq
\label{51'} x_{\rm null}=x_{\rm dual},
\eeq
i.e. the null vector is self-dual in the case of isometric $A^*$. 
Eq.  \eqref{51'}  can be used to  construct
the null vector from  $A_{I_c}A^*_{I_c}$ by the power method. 

Let $\mathbf{1}_{c}$ be the characteristic function of the complementary index $I_c$
with $|I_c|=\gamma N$.   The default choice for $\gamma$ is the  median value $\gamma=0.5$.
\begin{algorithm}
\SetKwFunction{Round}{Round}

\textbf{Random initialization:} $x_{1}=x_{\rm rand}$
\\
\textbf{Loop:}\\
\For{$k=1:k_{\textup{max}}-1$}
{
$x'_{k}\leftarrow A (\mathbf{1}_{c}\odot A^*x_k)$;\\
$x_{k+1}\leftarrow [x_k^{'}]_\cX/\|[x_k^{'}]_\cX\|$
}
{\bf Output:} {$\hat x_{\textup{dual}}=x_{k_{\textup{max}}}$.}
\caption{\textbf{The  null initialization}}
\label{null-algorithm}
\end{algorithm}

{ For isometric $A^*$, it is natural to define
\beq
\label{2.4}
x_{\rm null}=\alpha\|b\| \cdot \hat x_{\rm dual},\quad \alpha= {\hat x_{\rm dual}^*x_0\over
| \hat x_{\rm dual}^*x_0|}\eeq
where $\hat x_{\rm dual}$ is the output of Algorithm 1. 
As shown in Section \ref{sec:num} (Fig. \ref{fig:noise0}), 
the null vector is remarkably stable with respect to
noise in $b$. }

%


\subsection{Non-isometric $A^*$} When $A^*$ is non-isometric such as the standard Gaussian random matrix (see below),  the power method is still applicable with  the following modification.

For a full rank $A$, let $A^*=QR$ be  the QR-decomposition of $A^*$ where
$Q$ is isometric and $R$ is a full-rank, upper-triangular square matrix. Let $z=Rx$, $z_0=Rx_0$ and $z_{\rm null}=Rx_{\rm null}$. Clearly, $z_{\rm null}$ is the null vector for the isometric phase retrieval
problem $b=|Q z|$ in the sense of \eqref{2.3}. 

Let $I$ and $I_c$ be the index sets as above.
Let 
\beq
\hat z=\hbox{\rm arg}\max_{\|z\|=1} \|Q_{I_c}z\|.
\eeq
{ Then
\[
x_{\rm null}=\alpha\beta R^{-1}\hat  z
\]
where $\alpha$ is the optimal phase factor and 
\[ \beta={ \|x_0\|\over \|R^{-1} \hat z\|}
\]
may be an  unknown parameter in the non-isometric case.
As pointed out above, when $x_{\rm null}$ with an arbitrary parameter $\beta$ is used as initialization of phase retrieval, the first iteration of AP would recover the true value of $\beta$ as AP is totally independent of any real constant factor. }

\subsection{The spectral initialization}\label{sec:spectral}
Here we compare the null initialization with the spectral initialization used in \cite{CLS2} and the truncated spectral initialization used in \cite{truncatedWF}.

\begin{algorithm}[h]
\SetKwFunction{Round}{Round}
\textbf{Random initialization:} $x_1=x_{\rm rand}$\\
\textbf{Loop:}\\
\For{$k=1:k_{\textup{max}}-1$}
{
$x_k'\leftarrow A(|b|^2\odot A^*x_k);$\\
$x_{k+1}\leftarrow [x_k^{'}]_\cX/\|[x_k^{'}]_\cX\|$;
}
{\bf Output:} $x_{\rm spec}=x_{k_{\rm max}}$.
\caption{\textbf{The spectral initialization}}
\label{spectral-algorithm}
\end{algorithm}

The key difference between Algorithms 1 and 2 is the different weights used in
step 4 where  the null initialization uses $\mathbf{1}_{c}$ and  the spectral vector
method uses $|b|^2$ (Algorithm 2). 
The truncated spectral initialization uses  a still  different weighting 
\begin{equation}\label{truncated_version}
x_{\textup{t-spec}}=\textup{arg}\underset{\|x\|=1}{\textup{max}}
\| A
\left(\mathbf{1}_\tau\odot |b|^{2}\odot A^*x \right)\|
\end{equation}
where 
$\mathbf{1}_\tau$ is the characteristic function of 
the set
\[
\{i: |A^* x(i)|\leq \tau {\|b\|}\}
\]
with an adjustable parameter $\tau$. Both $\gamma$ of Algorithm 1 and $\tau$ of \eqref{truncated_version} can be 
optimized  by tracking and minimizing  the residual $\|b-|A^* x_k|\|$.

As shown in the numerical experiments in Section \ref{sec:num} (Fig. \ref{fig:initials} and~\ref{fig:initials_2masks}), the choice of weight  significantly
affects  the quality of initialization, with the null initialization as  the best performer
(cf.  Remark \ref{rmk5.2}). 

Moreover, because the null initialization depends only on the choice of the index set $I$ and not explicitly on $b$, the method is noise-tolerant and performs well with noisy data (Fig. \ref{fig:noise0}). 
\subsection{ Gaussian random measurement}\label{sec:Gaussian}

Although we are unable to provide a rigorous justification of
the null initialization in the Fourier case, we shall do so for 
 the complex Gaussian case $A=\Re (A)+i \Im(A)$, where  the entries of $\Re(A),\Im(A)$ are  i.i.d. standard normal random variables. 
The following error bound is in terms of
the closely related error  metric 
 \beq
 \label{52''}  \|x_0x_0^* -x_{\rm null}x_{\rm null}^*\|^2 
=2\|x_0\|^4- 2|x_0^*\xnul|^2
\eeq 
which has the advantage of being independent of the global phase factor. 

\begin{thm} \label{Gaussian}
Let $A\in \IC^{n\times N}$ be  an  i.i.d. complex standard Gaussian matrix. Suppose 
\beq
\label{53'}
\sigma:={|I|\over N}<1,\quad \nu={n\over |I|}<1.
\eeq
Then for any $\ep\in (0,1),\delta>0$ and $t\in (0, \nu^{-1/2}-1)$ 
 the following error bound 
\beq\label{error}
\|x_0x_0^* -x_{\rm null}x_{\rm null}^*\|^2
&\le& \lt( \left(\frac{2+t}{1-\epsilon} \right) \sigma+\ep \lt(-2\ln (1-\sigma) +\delta\rt)\rt){ 2\|x_0\|^4\over \left(1-(1+t)\sqrt{\nu}\right)^{2}}
 \eeq
holds with probability at least \beq
\label{prob}
&& 1-2\exp\left(-{N}{\delta^2 e^{-\delta}|1-\sigma|^2/2} \right)-
\exp(-{2}  \lfloor |I| \epsilon \rfloor^2/N )-Q
 \eeq
where $Q$ has the asymptotic upper bound 
 \beq\label{a.6}
 2 \exp\lt\{-c\min \lt[{e^2t^2 \over 16} \lt(\ln \sigma^{-1}\rt)^2 {|I|^2/N},~{et\over 4}|I|\ln\sigma^{-1}\rt]\rt\}, \quad\sigma \ll 1,\eeq
with  an absolute constant $c$. 
\end{thm}
\begin{rmk}
To unpack the implications of Theorem \ref{Gaussian}, consider the following asymptotic:
With $\ep$ and $ t $  fixed,  let 
\[
n\gg 1,\quad  \sigma={|I|\over N}\ll 1,\quad {|I|^2\over N}\gg 1, \quad\nu={n\over |I|}<1.
\]
We have
\beq
\label{5.8.1}
\|x_0x_0^* -x_{\rm null}x_{\rm null}^*\|^2\le c_0 \sigma\|x_0\|^4
\eeq
with probability at least
\[
1-c_1 e^{-c_2 n}-  c_3 \exp\lt\{-c_4\lt(\ln\sigma^{-1}\rt)^2 {|I|^2/ N} \rt\}
\]
  for moderate constants  $c_0, c_1, c_2,c_3,c_4$.
  
  To compare with the asymptotic regimes of \cite{CLS2} and \cite{truncatedWF}
  let us set $\nu<1$ to be a constant and  $N=Cn$ with a sufficiently large constant $C$. Then \eqref{5.8.1} becomes
\beq
\label{5.8.2}
\|x_0x_0^* -x_{\rm null}x_{\rm null}^*\|^2\le {c_0\over C\nu} \|x_0\|^4, 
\eeq
which is arbitrarily small with a sufficiently large constant $C$, with probability  close to 1 exponentially in $n$. 

{ In comparison, the performance guarantee for the spectral initialization (\cite{CLS2}, Theorem 3.3)  assumes 
$N=O(n\log n)$ for the same level of accuracy guarantee  with a success probability
less than $1-8/n^2$. On the other hand, the performance guarantee for the truncated
spectral vector  is comparable to Theorem \ref{Gaussian} in the sense that error bound like \eqref{5.8.2} holds true for the truncated spectral vector with $N=Cn$ and probability
exponentially close to 1 (\cite{truncatedWF}, Proposition 3). 

We mention by passing  that the initialization by Resampled Wirtinger Flow (\cite{CLS2}, Theorem 5.1) requires in practice a large number of coded diffraction patterns and
does not apply to the present set-up, 
so we do not consider  it further.}\\

 \label{rmk5.2}
\end{rmk}

The proof of Theorem  \ref{Gaussian} is given in Appendix A.

\section{AP}\label{sec:pap}

First we introduce some notation and convention that are frequently used in the subsequent analysis.

The vector space $\IC^n=\IR^n\oplus_\IR i\IR^n$ is 
isomorphic to
$\IR^{2n}$ via the map 
\begin{equation}\label{51} G(v):=\left[
\begin{array}{c}
\Re(v)     \\
 \Im(v)  
\end{array}
\right],\quad \forall v \in \IC^{n}\end{equation} 
and endowed with the real inner product
\[
\langle u, v\rangle :=\Re(u^*v)=G(u)^\top G(v),\quad u,v\in \IC^n.
\]
We say that $u$ and $v$ are (real-)orthogonal to each other (denoted by $u\perp v$)  iff $\langle u,v\rangle=0$. The same isomorphism exists between $\IC^{N}$ and $\IR^{2N}$. 

Let $ y\odot y'$ and $y/y'$ be the component-wise multiplication and division between  two vectors $y,y'$, respectively. For any $y\in \IC^N$ define the phase vector $\om\in \IC^N$ with  $\om(j)=\z(j)/|\z(j)|$ where $|\z(j)|\neq 0$.
When $|\z(j)|=0$ the phase can be assigned any value in $[0,2\pi]$. 
For simplicity, we set the default value $\z(j)/|\z(j)|=1$ whenever the denominator vanishes. 

It is important to note that for the measurement schemes \eqref{one} and \eqref{two}, the mask function by assumption is an array of  independent, continuous random variables
and so is $y_0=A^* x_0$.  Therefore  $b=|y_0|$ almost surely vanishes  nowhere. 
However, we will develop the AP method without assuming this fact and without specifically appealing to the structure of the measurement schemes \eqref{one} and \eqref{two} unless stated otherwise.

Let $A^*$ be any $N\times n$ matrix, $b=|A^* x_0|$ and 
\beq\label{3'}
F(x)&=&\half \||A^* x|-b\|^2  = \frac{1}{2}\|A^* x\|^2-\sum_{j\in J} b(j) |a_j^* x|+\frac{1}{2}\|b\|^2
\eeq   
where
 \[J:=\{j: b(j)>0\}.\]
  As noted above, for our measurement schemes \eqref{one} and \eqref{two}, 
 $J=\{1,2,\cdots, N\}$ almost surely. 
  
{   In view of \eqref{3'}, the only possible hinderance to differentiability for $F$ is the sum-over-$J$ term. Indeed, we have the following result. 

\begin{prop}\label{C2} The function $F(x)$ is infinitely differentiable in the open set 
\beq
\label{3.3'}
\{x\in \IC^n: |a_j^*x|> 0, \quad \forall j\in J\}. 
\eeq
In particular, 
for an isometric $A^*$,  $F(x)$ is  infinitely  differentiable  in the neighborhood of $x_0$ defined by \beq\label{4.4}
 \|x_0-x\|<\min_{j\in J} b(j).
 \eeq
\end{prop}
\begin{proof}
Observe that 
\begin{eqnarray*}
 |a_j^* x|&=& |a_j^* x_0-a_j^*(x_0-x)|\ge b(j)-|a_j^*(x_0-x)|\ge b(j)-\|x-x_0\|,
\end{eqnarray*}
and hence $|a_j^* x|> 0$ if 
$\|x_0-x\|<b(j)$. The proof is complete. \end{proof}
 }
 
Consider the smooth  function 
\beq
\label{3} f(x,u)&=&\half \| A^* x-u\odot b\|^2=\frac{1}{2}\|A^* x\|^2-\sum_{j\in J} \Re(x^*a_jb(j)u(j))+\frac{1}{2}\|b\|^2
\eeq
where $x\in \IC^n$ and
\beq
u\in U:=\{(u(i))\in \IC^N: |u(i)|=1,\,\, \forall i\}.\label{4}
\eeq
We can write 
\beq
\label{5}
F(x)&=&\min_{u\in U} f(x,u)
\eeq
which has many minimizers 
if $x^*a_j b(j)=0$ for some $j$. We select by convention the minimizer
\beq\label{3.6}
 u={A^* x\over |A^* x|}.
 \eeq
 
Define the complex gradient 
 \beq\label{3.6'}
\nabla_x f(x,u)&:= &
{\partial f(x,u)\over \partial \Re(x)}+i {\partial f(x,u) \over \partial \Im(x)}= AA^*x- A(u\odot b)
\eeq
and consider the alternating minimization
  procedure
\begin{eqnarray}\label{ERstep1}
u^{(k)}&=&\arg\min_{u\in U} f(x^{(k)},u),\\
x^{(k+1)}&=&\arg\min_{x\in \cX} f(x,u^{(k)})\label{ERstep2}
\end{eqnarray}
each of which is a least squares problem.

By \eqref{3.6} and \eqref{3.6'}, the minimizer \eqref{ERstep2} is given by
\beq\label{3.9}
x^{(k+1)}=(A^*)^\dagger(u^{(k)}\odot b),\quad u^{(k)}={A^* x^{(k)}\over |A^* x^{(k)}|}\eeq
where 
\[
(A^*)^\dagger= (AA^*)^{-1} A
\]
is the pseudo-inverse. 

\commentout{Equivalently, let $[\cdot ]_\cX$ denote the projection on the set $\mathcal{X}$ with respect to the Euclidean norm:
\begin{equation} z=[x]_\cX \textrm{ if $z\in\mathcal{X}$ and } \|z-x\|\le \|y-x\|,  \; \forall y\in \mathcal{X}.
\end{equation}
}

Eq. \eqref{3.9} can be written as the fixed point iteration
 \begin{equation}\label{papf}
x^{(k+1)}=\cF(x^{(k)}),\quad \cF(x)=\lt[(A^*)^{\dagger}\lt(b\odot \frac{A^* x}{|A^* x|}\rt)\rt]_\cX. 
\end{equation}
 In the one-pattern case, \eqref{papf} is exactly Fienup's Error Reduction algorithm \cite{Fie82}.  
 
The AP map \eqref{papf} can be formulated as the projected gradient method \cite{Goldstein, Levitin}. In the small neighborhood of $x_0$ where $F(x)$ is smooth (Proposition \ref{C2}), we have
\beq
\nabla F(x) &= &\nabla_x f(x, u)=AA^* x-A(b\odot u),\quad u={A^*x\over |A^* x|}\label{subdiff}
\eeq
and hence
\beq
\label{gp}
\cF(x)=\lt[ x- (AA^*)^{-1} \nabla F(x)\rt]_\cX. 
\eeq  
{ Where $F(x)$ is not differentiable, \eqref{subdiff} is an element of
the subdifferential of $F$. Therefore, the AP map \eqref{papf} can be viewed as
the generalization of the projected gradient method to the non-smooth setting. }

 The object domain formulation \eqref{papf} is equivalent to the Fourier domain formulation \eqref{fap} by the change of variables  $y=A^*x$ and letting
  \[
P_1 y= A^* [(A^*)^\dagger  y]_\cX,\quad P_2 y= b\odot {y\over |y|}.
 \]

 We shall study  the following three versions of AP. The first is the Parallel AP (PAP)
  \beq
\cF(x)=\label{3.11}(A^*)^\dagger\lt(b \odot {A^*x \over |A^*x|}\rt)
  \eeq
 to be applied to the two-pattern case. The second is the real-constrained AP (RAP)
   \beq
\cF(x)=\label{3.12}
  \lt[(A^*)^\dagger\lt(b \odot {A^*x\over |A^*x|}\rt)\rt]_{\cX},\quad \cX=\IR^n,~\IR^n_+
   \eeq
   to be applied to the one-pattern case. 
   
  The third is the Serial AP defined as follows. 
  Following \cite{rpi} in the spirit of Kaczmarz, we partition the measurement matrix and
the data vector into parts and treat them sequentially.

Let $A^*_l, b_l, l=1,2,$ be a partition of the measurement matrix and data, respectively, 
as  
  \[
  A^*=\lt[\begin{matrix} A_1^*\\ A^*_2\end{matrix}\rt],\quad b=\lt[\begin{matrix} b_1\\ b_2\end{matrix}\rt]
\]
with
 \[
 b_l=|A_l^* x_0|,\quad l=1,2.
 \]

 Let $y\in \IC^N$
be written as $y=[y_1^\top,y_2^\top]^\top$.  
Instead of \eqref{feas}, we now formulate the phase retrieval problem as the following 
feasibility problem
 \beq
 \label{feas'}
\hbox{Find}\quad  \hat y\in  \cap_{l=1}^2\lt(A_l^*\cX \cap \cY_l\rt),\quad \cY_l:=  \{y_l: |y_l|=b_l\}.
 \eeq
As  the projection onto the non-convex set $A_l^*\cX \cap \cY_l$ is not explicitly known,
we use the approximation instead 
 \begin{equation}\label{serial}
\cF_l(x)=(A^*_l)^\dagger\lt(b_l\odot \frac{A^*_l x}{|A^*_l x|}\rt),\quad l=1,2,
\end{equation}
and consider the Serial AP (SAP) map
    \beq
  \cF(x)=\cF_2(\cF_1(x)).  \label{3.13} 
  \eeq
In contrast,    the PAP map \eqref{3.11} 
\begin{equation}\label{para}
\cF(x)=A\lt(b\odot \frac{A^* x}{|A^*x|}\rt) =\cF_1(x)+\cF_2(x)
\end{equation}
 is the sum of $\cF_1$ and $\cF_2$. 
Note that $x_0$ is a fixed point of both $\cF_1$ and $\cF_2$.

 \section{Fixed points}\label{sec:fixed}
 Next we study the fixed points of PAP and RAP. Our analysis does not
 extend to the case of SAP. 
 
{Following \cite{rpi} we consider  the the generalized AP (PAP) map
 \beq
 \cF_{u}(x):=\label{3.7.2}
  \lt[A\lt(b\odot  u \odot {A^*x \over |A^*x|}\rt)\rt]_\cX,\quad \cX=\IC^n,~\IR^n,~\IR^n_+
 \eeq
 where 
 \beq
\label{fix2}
u\in U,\quad u(j)=1,\quad \hbox{whenever}~A^*x_*(j)\neq 0.
\eeq
We call $x_*$ a {\bf fixed point} of AP if there exists 
\[
u\in U=\{
u=(u(i))\in \IC^N: |u(i)|=1,\,\, \forall i\}\
\] satisfying \eqref{fix2} and 
\beq
\label{fixed}
x_*=\cF_u(x_*),
\eeq
\cite{rpi}. 
 In other words, the definition \eqref{fixed} allows flexibility of phase
 where $A^*x_* $ vanishes.}
 
First we identify any limit point of the AP iterates  
with a fixed point of AP. \\

\begin{prop}\label{Cauchy} The AP iterates  $x^{(k)}=\cF^{k}(x^{(1)})$ with any starting point $x^{(1)}$, where $\cF$ is given by \eqref{3.11} or \eqref{3.12},  is bounded
and 
every limit point  is a fixed point of AP in the sense \eqref{fix2}-\eqref{fixed}.\\
\end{prop}
\begin{proof}
Due to \eqref{5} , (\ref{ERstep1}) and (\ref{ERstep2}), \begin{equation} \label{11}
0\le F(x^{(k+1)})=f(x^{(k+1)},u^{(k+1)})\le f(x^{(k+1)},u^{(k)}) \le f(x^{(k)},u^{(k)})=F(x^{(k)}),\quad \forall k,\end{equation}
and hence 
 AP yields a non-increasing sequence
$\{F(x^{(k)})\}_{k=1}^\infty$.


For an isometric $A^*$,    \[
\nabla_x f(x,u) =x- A(u\odot b),
\]
and 
 \beqn
 \cF(x)&=& \lt[x- \nabla_x f(x,u)\rt]_\cX, \quad u={A^* x\over |A^* x|}. 
\eeqn
implying 
 \beq
 \label{12}
 x^{(k+1)}=[x^{(k)}-\nabla_x f(x^{(k)},u^{(k)})]_\cX. 
 \eeq
 
 Now by the convex projection theorem (Prop. B.11 of \cite{Np}).\begin{equation}\label{21}
\langle x^{(k)}-\nabla_x f(x^{(k)},u^{(k)})-x^{(k+1)}, x-x^{(k+1)} \rangle \le 0,  \quad\forall x\in \mathcal{X}
\end{equation}
Setting $x=x^{(k)}$ in  Eq.~(\ref{21}) we have
\begin{equation}\label{22}
\|x^{(k)}- x^{(k+1)}\|^2\le \langle \nabla_x f(x^{(k)},u^{(k)}), x^{(k)}- x^{(k+1)}\rangle.
\end{equation}
Furthermore, the descent lemma (Proposition~A.24, ~\cite{Np}) yields
\begin{equation}\label{20}
f( x^{(k+1)},u^{(k)})\le f(x^{(k)},u^{(k)})+ \langle  x^{(k+1)}-x^{(k)},  \nabla_x f(x^{(k)},u^{(k)})\rangle+\frac{1}{2} \| x^{(k+1)}-x^{(k)}\|^2.
\end{equation}

From  Eq.~(\ref{11}),    Eq.~(\ref{20}) and Eq.~(\ref{22}), we have
 \begin{eqnarray}\label{3.23}
F(x^{(k)})-F(x^{(k+1)})&\ge& f(x^{(k)},u^{(k)})-f(x^{(k+1)},u^{(k)})\\
&\ge&   \langle  x^{(k)}-x^{(k+1)}, \nabla_x f(x^{(k)},u^{(k)}) \rangle-\frac{1}{2} \| x^{(k+1)}-x^{(k)}\|^2\nn\\
& \ge& \frac{1}{2} \| x^{(k+1)}-x^{(k)}\|^2.\nn
\end{eqnarray}
As a nonnegative and non-increasing sequence, $\{F(x^{(k)})\}_{k=1}^\infty$ converges
and then \eqref{3.23} implies 
\beq
\label{23}
\lim_{k\to\infty}\|x^{(k+1)}-x^{(k)}\|=0. 
\eeq
By the definition of $x^{(k)}$ and the isometry of $A^*$, we have
\[\|x^{(k)}\|\le
\|A(b\odot u^{(k-1)})\|\le \|b\|,\]
and hence $\{x^{(k)}\}$ is bounded.
Let $\{x^{(k_j)}\}_{j=1}^\infty$ be a convergent subsequence and  $x_*$ its limit. 
Eq.  \eqref{23} implies that 
\[
\lim_{j\to\infty} x^{(k_j+1)}=x_*.
\]

If $A^* x_*$ vanishes nowhere, then $\cF$ is continuous at $x_*$.  
Passing to the limit in $\cF(x^{(k_j)})=x^{(k_j+1)}$
we get $\cF(x_*)=x_*$.  Namely, $x_*$ is a fixed point of $\cF$.

Suppose $a_l^* x_*=0$ for some $l$. By the compactness of the unit circle and 
further selecting a subsequence, still denoted by $\{x^{(k_j)}\}$, we have
\[
\lim_{j\to\infty} {A^* x^{(k_j)}\over |A^* x^{(k_j)}|}=A^* x_*\odot u
\]
for some $u\in U$ satisfying \eqref{fix2}. 
Now passing to the limit in $\cF(x^{(k_j)})=x^{(k_j+1)}$
we have \[
x_*=\cF_u(x_*)
\]   
implying that $x_*$ is a fixed point of AP.\commentout{
\[
\lt[ x^{(k_j)} -\nabla F( x^{(k_j)})\rt]_\cX=x^{(k_j+1)}.
 \]
we obtain 
\begin{equation}\lt[ x_* -\nabla F(  x_*)\rt]_\cX=x_*.\end{equation}
}
\end{proof}

 Since the true object is unknown, 
the following norm criterion is useful for distinguishing the phase retrieval solutions from the non-solutions among many coexisting fixed points. 
 
\begin{prop} \label{prop2.21} 
Let $\cF$ be the AP map \eqref{papf} with isometric $A^*$.  
If   a fixed point $x_*$ of AP in the sense \eqref{fix2}-\eqref{fixed} satisfies $\|x_*\|= \|b\|$, then
 $x_*$ is a phase retrieval solution  almost surely. 
On the other hand, if $x_*$ is not a phase retrieval solution, then $\|x_*\|<\|b\|$. 
\end{prop}
\begin{rmk}
If the isometric $A^*$ is specifically given by \eqref{two} or  \eqref{one}, then we can identify  any fixed point  $x_*$ satisfying the norm criterion  $\|x_*\|=\|b\|$
with the unique phase retrieval solution $x_0$ in view of Proposition \ref{prop:unique}. 
\end{rmk}
\begin{proof} By  the convex projection theorem (Prop. B.11 of \cite{Np})
\beq
\label{3.14}
\|[v]_\cX\|\le \|v\|, \quad \forall v\in \IC^n
\eeq
where the equality holds if and only if $v\in \cX$.
Hence 
\beq
\label{317}\|x_*\|&=& \lt\| \lt[A\lt(\frac{A^* x_*}{|A^* x_*|}\odot b\odot u\rt)\rt]_\cX \rt\| \\
\nn&\le & \lt\| A\lt(\frac{A^* x_*}{|A^* x_*|}\odot b\odot u\rt) \rt\|\\
&\le & \lt\|\frac{A^* x_*}{|A^* x_*|}\odot b\odot u\rt\|=\|b\|.\nn
 \eeq
Clearly $\|x_*\|=\|b\|$ holds if and only if  both inequalities in Eq.~(\ref{317})  are equalities. The second inequality is an equality only when     \begin{equation}\label{11'}
 \frac{A^* x_*}{|A^* x_*|}\odot b\odot u= A^* z \quad \textrm{ for some $z \in \IC^n$}. 
 \end{equation} By Eq.~(\ref{3.14}) and (\ref{11'}) the first inequality in Eq.~(\ref{317}) becomes an equality only when $
 z\in \cX$. 
 
 Since $AA^*=I$ the fixed point equation \eqref{fixed} implies  $z=x_*$ 
 and    \[
 \frac{A^* x_*}{|A^* x_*|}\odot b\odot u=A^* x_*. 
 \]
Thus $b=|A^* x_*|$.
\end{proof}

 \section{Parallel AP}\label{sec:loc}


Define 
  \beq B_x&=&A\, \diag\lt[\frac{A^* x}{|A^*x|}\rt]\label{25}\\
  \label{B} \mathcal{B}_x&=&\left[
\begin{array}{c}
\Re(B_x)     \\
 \Im(B_x)
\end{array}
\right].\eeq
When $x=x_0$, we will drop the subscript $x$ and write simply $B$ and $\cB$.

\commentout{
From \eqref{51} we  have
\beq\label{53}
G(B_x^*u)=\lt[\begin{matrix}
\cB_x^\top G(u)\\
           \cB_x^\top G(-i u)
           \end{matrix}\rt],\quad u\in \IC^n. 
           \eeq
           }
           
Whenever $F(x)$ is differentiable at $x$, we have as before  
\beq
\label{4.5}\nabla F(x) &= &
{\partial F(x)\over \partial \Re(x)}+i {\partial F(x) \over \partial \Im(x)}\\
&=&A(A^* x-b\odot u),\quad u={A^*x\over |A^* x|}\nn
\eeq
and 
  \beq
  \nabla^2 F(x)\zeta &:=& \nabla \langle \nabla F(x),\zeta\rangle\\
 \nn &=& {\partial \langle \nabla F(x),\zeta\rangle\over \partial \Re(x)}+i {\partial \langle \nabla F(x),\zeta\rangle\over\partial  \Im(x)},\quad\forall \zeta\in \IC^n. 
    \eeq
 
\begin{prop}  \label{prop4.1} {Suppose $|a_j^*x|> 0$ for all $j\in J=\{i:b_i>0\}$ (i.e.  $F(x)$ is smooth at $x$ by Proposition \ref{C2}).  } For  all $\zeta \in \IC^n$, we have 
\begin{equation} \langle \nabla F(x),\zeta\rangle =\Re(x^{*}\zeta)-b^{\top} \Re(B_x^* \zeta),\label{28}\end{equation}
and \beq\label{5.6}
\langle \zeta, \nabla^2 F(x) \zeta\rangle &=&\|\zeta\|^2-\langle  \Im\left(  B_x^* \zeta \right), \rho_x\odot \Im\left(  B_x^* \zeta \right)\rangle\\
&=&\|\zeta\|^2-\langle     \mathcal{B}_x^\top G(-i\zeta) , \rho_x \odot \mathcal{B}_x^\top G(-i\zeta )\rangle\nn
\eeq
with
  \[
  \rho_x(j)=\lim_{\ep \to 0^+}\frac{b(j)}{\ep+|a_j^* x|},\quad j=1,\cdots,N. 
  \]


\end{prop}
\begin{proof} Rewriting \eqref{3'} as  \beq\label{FF}
F(x)
&=&\frac{1}{2}\|  A^* x\|^2- \sum_{j\in J} f_j(x)+\frac{1}{2}\|b\|^2,\quad f_j(x):=b(j)|a_j^*x|,
\eeq
we analyze the derivative of each term on the right hand side of \eqref{FF}.

Since $AA^{*}=I$, the gradient and the Hessian of   $\|  A^* x\|^2/2$ are $x$ and  $I$, respectively.

For $f_j$,  we have   Taylor's expansion
\begin{equation}\label{5.8} f_j(x+\ep\zeta) =f_j(x)+\ep \langle \nabla f_j(x), \zeta\rangle + \frac{\ep^2}{2}  \langle \zeta, \nabla^2 f_j(x)\zeta\rangle  + O(\ep^3)
\end{equation}
where 
\begin{equation}
\langle \nabla f_j(x),\zeta\rangle =
 \frac{b(j)}{|a_j^* x|} \langle a^*_j x, a^*_j \zeta\rangle,\quad j\in J
\end{equation}
and
 \beq\label{5.9}
 \langle\zeta,  \nabla^2 f_j(x)\zeta  \rangle &=& \frac{b(j)}{|a_j^*x|}  \lt|{\Re(a_j^* x) \over |a_j^*x|} \Im( a_j^*\zeta)-{\Im( a_j^*x)\over|a_j^*x|} \Re (a_j^*\zeta)\rt|^2,\quad j\in J. 
\eeq
Observe  that
\beq\nn
 \langle {a^*_j x\over |a_j^* x|}, a^*_j \zeta\rangle=\Re(B^*_x\zeta)(j),\quad j\in J
\eeq
and 
\[
{\Re(a_j^* x) \over |a_j^*x|} \Im( a_j^*\zeta)-{\Im( a_j^*x)\over|a_j^*x|} \Re (a_j^*\zeta)=\Im(B^*_x\zeta)(j)=\cB^\top_x G(-i\zeta)(j), \quad j\in J
\]
which, together with \eqref{5.8} and \eqref{5.9},  yield the desired results \eqref{28} and \eqref{5.6}. 
\end{proof}


Next we investigate the conditions under which  $\nabla^ 2 F(x_0)$ is positive definite.

\subsection{Spectral gap}

 Let $\lambda_1\ge \lambda_2\ge \ldots\ge \lambda_{2n}\ge \lambda_{2n+1}=\cdots=\lambda_{N}=0$ be the singular values of $\mathcal{B}$ with the corresponding right singular vectors
$\{\bv_k\in \IR^{N}\}_{k=1}^{N} $ and left singular vectors $\{{\bu}_k\in \IR^{2n}\}_{k=1}^{2n}$.

%
%
%

\begin{prop}\label{prop4.2}
We have $\lambda_1=1, \lambda_{2n}=0$,  $\eta_1=|A^* x_0|$ and 
\[
\bu_{1}=G(x_0)=\lt[\begin{matrix} \Re(x_0)\\
\Im(x_0)\end{matrix}\rt],  \quad \bu_{2n}=G(-ix_0)=\lt[\begin{matrix} \Im(x_0)\\
-\Re(x_0)
\end{matrix}\rt]
\]
\end{prop}
\begin{proof} Since
 \beq\nn
&& B^* x=\Om^* A^* x,\quad \Om= \diag\lt[\frac{A^* x_0}{|A^*x_0|}\rt]\eeq
we have 
\beq
& &\Re[B^* x_0]=\mathcal{B}^\top \xi_{1}=|A^* x_0|, \quad \Im[B^*x_0]=\mathcal{B}^\top \xi_{2n}=0  \label{56}\eeq
 and hence the results.
\end{proof}

\begin{prop}\label{cor5.5}
 \beq 
\label{63} \lambda_2&= &\max\{ \|\Im[B^* u]\|:  {u\in \IC^n,  iu\perp x_0}, \|u\|=1 \}\\
&=&\max \{\|\cB^\top u\|: {u\in \IR^{2n},  u\perp \xi_1}, \|u\|=1\}. \nn
 \eeq
\end{prop}

\begin{proof}
Note that 
\[
\Im[B^* u]=\cB^\top G(-i u).
\]
The orthogonality condition $iu\perp x_0$ is equivalent to
\[
G(x_0)\perp G(-iu).
\]
Hence, by
 Proposition \ref{prop4.2},  $\bu_2 $ is the maximizer of the right hand side of \eqref{63}, yielding the desired value $\lambda_2$.

\end{proof}

We recall the spectral gap property, proved in \cite{DR-phasing}, that is a key to
local convergence of the one-pattern and the two-pattern case.

  \begin{prop}\label{cor5.2} \cite{DR-phasing}
Suppose $x_0\in \IC^n$ is rank $\ge$ 2. For $A^*$ given by  \eqref{two} with independently and continuously distributed mask phases,  we have $\lambda_2 <1$
 with probability one. 
 \end{prop}

\begin{prop} \label{prop4.5} Let  \begin{equation} \label{38} \lambda_2(x)=
 \max\{ \|\Im (B_x^* u)\|: {u\in \IC^n, \langle u, x\rangle=0, \|u\|=1
 }\}.\end{equation}
 Let $\gamma$
be a convex combination of $x$ and $x_0$ with $\langle x_0, x\rangle >0$.
 Then
\begin{equation}\| \Im (B_\gamma^* (x-x_0))\|\le \lambda_2(\gamma)   \|x-x_0\|  .\end{equation}
\end{prop}

\begin{proof}
   Since $\langle x_0, x\rangle >0$, 
\begin{equation}
c_1:=\|\gamma\|^{-2} \langle\gamma, x_0\rangle >0,\quad c_2:=\|\gamma\|^{-2}\langle \gamma, x\rangle >0
\end{equation}
and we can write the orthogonal decomposition 
\begin{equation}
x_0=c_1 \gamma+\gamma_1,\quad
x=c_2 \gamma+\gamma_2
\end{equation}
with some vectors $\gamma_1,\gamma_2$ satisfying  $\langle\gamma_1, \gamma\rangle =\langle \gamma_2, \gamma\rangle =0$.

 By \eqref{25}, 
 \[
 \Im(B_\gamma^* \gamma)=\Im(|A^*\gamma|)=0
 \]
 and hence 
\beqn
\Im(B_\gamma^* (x-x_0)) &=& \Im(B_\gamma^*(\gamma_2-\gamma_1))
\eeqn
from which it follows that 
\beqn
\|x-x_0\|^{-1}\|\Im(B_\gamma^* (x-x_0))\|&\le& \| \gamma_2-\gamma_1 \|^{-1}\|\Im(B_{\gamma}^*( \gamma_2-\gamma_1))\|\le \lambda_2(\gamma)
\eeqn
by the definition \eqref{38}. 
\end{proof}

\subsection{Local convergence}
We state the local convergence theorem for arbitrary isometric $A^*$, not necessarily given by the Fourier measurement. 
\begin{thm} \label{thm1} (PAP) For any isometric $A^*$, let $b=|A^* x_0|$ and
$\cF$ be given by \eqref{3.11}. 
Suppose $\lambda_2<1$ where $\lambda_2$ is given by \eqref{63}. 

For any given $0<\ep< 1-\lamb_2^2$, if $x^{(1)}$ is sufficiently close to $x_0$ then with probability one the AP iterates $x^{(k+1)}=\cF^k(x^{(1)})$ converge to $x_0$  geometrically after global phase adjustment, i.e.  
\begin{equation}\label{41'}
{\|\alpha^{(k+1)}x^{(k+1)}- x_0\|}\le (\lambda^2_2+\ep) {\|\alpha^{(k)}x^{(k)}- x_0\|},\quad \forall k
\end{equation}
where
 $\alpha^{(k)}:=\hbox{\rm arg}\min_{\alpha} \{ \|\alpha x^{(k)} -x_0\| :  |\alpha|=1\}$. 
\end{thm}
\begin{proof} By Proposition \ref{C2} and the projected gradient formulation \eqref{gp},  we have
\beq
\cF(x)&=& x-\nabla F(x). \nn
\eeq

From the definition of $\alpha^{(k+1)}$, we have 
\beq
\label{41}\|\alpha^{(k+1)}x^{(k+1)}-x_0\|
&\le& \| \alpha^{(k)} x^{(k+1)}-x_0\| \\
&\le &  \| \alpha^{(k)} x^{(k)}-\nabla F(\alpha^{(k)} x^{(k)})-x_0+\nabla F(x_0)\| \nn
\eeq

Let $ g(x)=x-\nabla F(x)$ and $\gamma(t)=x_0+t(x-x_0)$.  
By the  mean value theorem, 
\begin{equation}\label{42.1}
g(x)-g(x_0)=\int_{0}^{1}\left[I-\nabla^{2} F(\gamma(t))\right] (x-x_0)dt
\end{equation}
and hence with $x=\alpha^{(k)}x^{(k)}$
the right hand side of \eqref{41} equals 
\beqn
&&\|\int_{0}^{1}(I-\nabla^2F(\gamma(t)))(\alpha^{(k)}x^{(k)}-x_0)dt\|\\
&=&\|
\int_{0}^{1}\mathcal{B}_{\gamma(t)} (\rho_{\gamma(t)} \odot \mathcal{B}_{\gamma(t)}^{\top} G(-i(\alpha^{(k)}x^{(k)}- x_0)))dt
\|
\eeqn
by Proposition \ref{prop4.1}, and  is  bounded  by 
\beq\label{49'}
&& \|
\int_{0}^{1}\mathcal{B}_{\gamma(t)}\Big( (\rho_{\gamma(t)}-{\bf 1}_J) \odot \mathcal{B}_{\gamma(t)}^{\top} G(-i(\alpha^{(k)}x^{(k)}-x_0))\Big)dt
\|\\
&&
+\|
\int_{0}^{1}\mathcal{B}_{\gamma(t)} {\bf 1}_J\odot \mathcal{B}_{\gamma(t)}^{\top} G(-i(\alpha^{(k)}x^{(k)}- x_0))dt\| \nn
\eeq
where ${\bf 1}_J$ is the indicator of $J=\{j: b_j>0\}$. 

Since 
\[
\|\alpha x^{(k)} -x_0\|^2=\|x^{(k)}\|^2+\|x_0\|^2-2\langle \alpha x^{(k)},x_0\rangle,
\]
we have 
\[
\langle \alpha^{(k)} x^{(k)},x_0\rangle>0,\quad \forall k
\]
and hence, by Proposition \ref{prop4.5}, 
\beq
\label{43.1}\|
\cB_\gamma  \cB_\gamma^* G(-i(\alpha^{(k)}x^{(k)}-x_0))
\|\le \lambda^2_2(\gamma) \| \alpha^{(k)}x^{(k)}-x_0
\|\eeq so we can bound \eqref{49'} by 
\beqn
\lt(\sup_{t\in(0,1)}\|\rho_{\gamma(t)} -{\bf 1}_J\|_{\infty}+ 
\sup_{t\in (0,1)}  \lambda^2_2(\gamma(t))\rt)\| \alpha^{(k)}x^{(k)}-x_0\|.
\eeqn

For any $\ep>0$, if $x^{(1)}$ is sufficiently close to $x_0$, then by continuity   \begin{equation} \sup_{t\in (0,1)}  \lambda^2_2(\gamma(t))\leq \lambda^2_2+\ep/2,\quad \sup_{t\in(0,1)}\|\rho_{\gamma(t)} -{\bf 1}_J\|_{\infty}\leq \ep/2,  \end{equation} 
and we have from above estimate 
\[
\|\alpha^{(2)}x^{(2)}-x_0\|\leq (\lamb^2_2+\ep)  \|\alpha^{(1)}x^{(1)}-x_0\|.\]
By induction, we have
\[
\|\alpha^{(k+1)}x^{(k+1)}-x_0\|\leq (\lamb^2_2+\ep)  \|\alpha^{(k)}x^{(k)}-x_0\|\]
from which \eqref{41'} follows.

\end{proof}

\section{Real-constrained AP}\label{sec:one-pattern}

In the case of $x_0, x\in \IR^n$ (or $\IR^n_+$), we adopt the new definition
\beq \label{63'}\tilde  \lambda_2&:=&
 \max\{\|\Im (B^*) u\|:{u\in \IR^n, \langle u, x_0\rangle=0, \|u\|=1}\}
 \eeq
 which differs  from the definition  \eqref{38} of  $\lambda_2$  in that $u$ has all real components. Clearly we have $\tilde\lambda_2\leq \lambda_2$ of the one-pattern case.
  
From the isometry property of $B^*$ and that $u\in \IR^n$, it follows that 
 \beq
\tilde \lambda^2_2
 &=&1-\min\{\|\Re(B^*)u\|^2:{u\in \IR^n, \langle u,x_0\rangle=0,\|u\|=1}\}. \label{48.1} 
 \eeq
 By Proposition \ref{prop4.2} and $x_0\in \IR^n$, 
 \[
 \xi_1=\lt[\begin{matrix} x_0\\ 0
 \end{matrix}\rt]
 \]
 and hence $x_0$ is the leading singular vector of $\Re(B^*)$ over $\IR^n$.  Therefore, we can remove the condition  $ \langle u,x_0\rangle=0$   in \eqref{48.1} and write 
 \beq\label{49.1}
\tilde  \lambda^2_2 &=&1-\min_{u\in \IR^n\atop\|u\|=1}\|\Re(B^*)u\|^2\\
&=&\max_{u\in \IR^n\atop\|u\|=1}\|\Im(B^*)u\|^2\nn\\
&=&\|\Im(B^*)\|^2. \nn
\eeq
 
 The spectral gap property $\lambda_2<1$  holds even with just one coded diffraction pattern
 for any complex object.
 \begin{prop} \label{prop4.8} \cite{DR-phasing} Let $x_0\in \IC^n$  be rank $\ge$ 2. For $A^*$ given by \eqref{one} with independently and continuously distributed mask phases, 
 \beqn
\lambda_2= \max\{\|\Im[B_l^* u]\|: {u\in \IC^n,  iu\perp x_0}, \|u\|=1\}<1
\eeqn 
and hence $\tilde\lambda_2<1$  with probability one. 
\end{prop}

Following  {\em verbatim} the proof  of
Proposition \ref{prop4.5}, we have the similar result.

\begin{prop} \label{prop4.5'} Let  $x_0, x\in \IR^n$ (or $\IR^n_+$) with $\langle x_0, x\rangle >0$.  Let $\gamma$
be a convex combination of $x$ and $x_0$.
 Then
\begin{equation}\| \Im (B_\gamma^* (x-x_0))\|\le \tilde \lambda_2(\gamma)   \|x-x_0\|\end{equation}
where
\[
\tilde \lambda_2(\gamma) :=\max\{\|\Im (B_\gamma^*) u\|: {u\in \IR^n, \langle u, \gamma\rangle=0,\|u\|=1
 }\}.
 \]
\end{prop}

 
 The following convergence theorem is analogous to Theorem \ref{thm1}. 
  \begin{thm} \label{thm2} (RAP) 
For any isometric $A^*$, let $b=|A^* x_0|$ and
$\cF$ be given by \eqref{3.12}. 
Suppose $\tilde\lambda_2<1$ where $\tilde\lambda_2$ is given by \eqref{63'}.  
  
 For any given $0<\ep<1-\tilde\lamb_2^2$, if $x^{(1)}$ is sufficiently close to $x_0$ then with probability one the AP iterates $x^{(k+1)}=\cF^{k}(x^{(1)})$ converge to $x_0$  geometrically after global phase adjustment, i.e.  
\begin{equation}\label{41.2}
{\|\alpha^{(k+1)}x^{(k+1)}- x_0\|}\le (\tilde \lambda^2_2+\ep) {\|\alpha^{(k)}x^{(k)}- x_0\|},\quad \forall k
\end{equation}
where
 $\alpha^{(k)}:=\hbox{\rm arg}\min_{\alpha=\pm 1} \{ \|\alpha x^{(k)} -x_0\| \}$ and $\alpha^{(k)}=1$ if $x_0\in\IR^n_+$. 
\end{thm}
\begin{proof}

From the definition of $\alpha^{(k+1)}$, we have 
\beq
&&\|\alpha^{(k+1)}x^{(k+1)}-x_0\|
\le \| \alpha^{(k)} x^{(k+1)}-x_0\| \label{41.1}
\eeq

Recalling \eqref{gp}, we write 
 \begin{eqnarray*}
x^{(k+1)}=\lt[x^{(k)}-\nabla F(x^{(k)})\rt]_\cX. 
\end{eqnarray*}
By the properties of  linear projection,
 \beq
\alpha^{(k)}x^{(k+1)}
&=& \lt[\alpha^{(k)}x^{(k)}-\nabla F(\alpha^{(k)}x^{(k)})\rt]_\cX \eeq
and hence the right hand side of \eqref{41.1} equals 
\beq
\nn &&\|[ \alpha^{(k)}x^{(k)}-\nabla F(\alpha^{(k)} x^{(k)}) ]_\cX- [x_0-\nabla F(x_0)]_\cX\|\\
&\le& \| \alpha^{(k)} x^{(k)}-\nabla F(\alpha^{(k)} x^{(k)})-x_0+\nabla F(x_0)\|. 
\eeq

The rest of the proof follows {\em verbatim} that of Theorem \ref{thm1} from \eqref{42.1} onward, except with $\lambda_2$ replaced by $\tilde \lambda_2$. 
\end{proof}

%
%
%
%
%
%
%
%
%
%
%
%
%
%
%
%
%

\section{Serial AP}\label{sec:SAP}
To build on the theory of PAP, we 
assume, as for two coded diffraction patterns, $A= [A_1,A_2]$ where
$A_l^* \in \IC^{N/2\times n}$ are isometric
and let $b_l=|A_l^* x_0|\in \IR^{N/2}$.

By applying Theorem \ref{thm1} separately to $\cF_1$ and $\cF_2$, we get the following bound
\begin{equation}\label{58.2}
{\|\alpha^{(k+1)}x^{(k+1)}- x_0\|}\le ((\lambda_2^{(2)}\lambda^{(1)}_2)^2+\ep) {\|\alpha^{(k)}x^{(k)}- x_0\|},\quad \forall k,
\end{equation}
where \beqn
\lambda^{(l)}_2= \max\{\|\Im[B_l^* u]\|: {u\in \IC^n,  iu\perp x_0}, \|u\|=1\}, \quad B_l=A_l\diag\lt\{A_l^*x_0\over |A^*_l x_0|\rt\},
\eeqn
$  l=1,2.$
But we can do better. 

 Similar to the calculation in Proposition \ref{prop4.1}, the derivative $d\cF_l$ of $\cF_l$ in the notation of \eqref{51}, \eqref{25},\eqref{B} can be expressed as 
\beqn
G(d\cF_l \xi)&=&G(iB_l\Im(B_l^* \xi))\\
&=&\lt[\begin{matrix}
-\Im(B_l)\\
\Re(B_l)\end{matrix}\rt] \cB_l^\top  G(-i\xi),\quad \forall\xi\in \IC^n.
\eeqn
Equivalently, we have
\beqn
G(-id\cF_l \xi)&=&\cB_l\cB_l^{\rm T}  G(-i\xi),\quad \forall\xi\in \IC^n.
\eeqn
Hence, by the isomorphism  $\IC^{n}\cong \IR^{2n}$
via $G(-i\xi)$, we can represent the action of $d\cF_l$ on $\IR^{2n}$ by the real matrix
\beq
\cB_l\cB_l^\top=\lt[\begin{matrix}
\Re(B_l)\\
\Im(B_l)\end{matrix}\rt] \lt[
\Re(B_l^\top) \,\, \Im(B^\top_l)\rt] \label{59.2}
\eeq
and the action of $d(\cF_2\cF_1)$ by
\beqn
\cD:=\cB_2\cB^\top_2 \cB_1\cB^\top_1.\eeqn
Define 
\beq\label{60.2}
\|\cD\|_\perp&:=&\max \{\|\cD \xi\|: {\xi\in \IR^{2n}, \xi\perp \xi_1},\|\xi\|=1\}.
\eeq

We have the following bound. 
\begin{prop}\label{prop4.11}
\beq
\|\cD\|_\perp 
&\leq & (\lambda_2^{(2)}\lambda^{(1)}_2)^2\nn.
\eeq
\end{prop}
\begin{rmk}
By Proposition \ref{prop4.8}, $\lamb_2^{(l)}<1, l=1,2,$ and hence $\|\cD\|_\perp<1$. \end{rmk}
\begin{proof}
Since $\xi_1=G(x_0)$ is the fixed point for both $\cB_1\cB^\top_1$ and $\cB_2\cB_2^\top$,
the set $\{\xi\in \IR^{2n}: \xi\perp\xi_1\}$ is invariant under both. 
 Hence, by the calculation 
\beqn
\|\cB_2\cB^\top_2 \cB_1\cB^\top_1\xi\|&=&\|\cB_2\cB_2^\top \xi'\|,\quad \xi'=\cB_1\cB_1^\top \xi\\
&\leq&( \lamb_2^{(2)})^2 \|\xi'\|\\
&\leq&( \lamb_2^{(2)})^2 ( \lamb_2^{(1)})^2 \|\xi\|
  \eeqn
  the proof   is complete.  

\end{proof}

 We now prove the local convergence of SAP. 
\begin{thm} \label{thm3} (SAP) For any isometric $A^*$, let $b=|A^* x_0|$ and
$\cF$ be given by \eqref{3.13}. 
Suppose $\|\cD\|_\bot<1$ where $\|\cD\|_\bot$ is given by \eqref{60.2}.

For any given $0<\ep<1-\|\cD\|_\perp$, if $x^{(1)}$ is sufficiently close to $x_0$ then with probability one the AP iterates $x^{(k+1)}=\cF^{k}(x^{(1)})$ converge to $x_0$  geometrically after global phase adjustment, i.e.  
\begin{equation}\label{41.3}
{\|\alpha^{(k+1)}x^{(k+1)}- x_0\|}\le (\|\cD\|_\perp+\ep) {\|\alpha^{(k)}x^{(k)}- x_0\|},\quad \forall k
\end{equation}
where
 $\alpha^{(k)}:=\hbox{\rm arg}\min_{\alpha} \{ \|\alpha x^{(k)} -x_0\| :  |\alpha|=1\}$. 
\end{thm}
\begin{proof}

At the optimal phase $\alpha^{(k)}$ adjustment for $x^{(k)}$,  we have  
\[
\Im(x_0^*\alpha^{(k)} x^{(k)})=0
\]
and hence  \begin{equation} \langle  \alpha^{(k)} x^{(k)}-x_0 ,i  x_0\rangle=\langle  \alpha^{(k)} x^{(k)}, i  x_0\rangle =\Re(({\alpha^{(k)} x^{(k)} })^*i x_0)=0\end{equation}
which implies  that \[
u^{(k)}:=-i(\alpha^{(k)}x^{(k)}-x_0)
\]
 is orthogonal to the leading right singular vector $\bu_1=G(x_0)$ of $\cB^*_l, l=1,2$:
\beq
\bu_1\perp G(u^{(k)}),\quad\forall k\label{74}
\eeq
cf. Proposition \ref{prop4.2}.

We have for $k=1,2,3,\cdots$ 
\beqn
\|\alpha^{(k+1)}\cF_2\cF_1(x^{(k)})-x_0\|&\leq&\|\alpha^{(k)}\cF_2\cF_1(x^{(k)})-x_0\|\\
&=&\|\cF_2\cF_1(\alpha^{(k)}x^{(k)})-\cF_2\cF_1(x_0)\|\\
&=& \|\cD G(u^{(k)})\|+o( \|u^{(k)}\|)\\
&\leq&\max_{\xi\perp \xi_1\atop \|\xi\|=1} \|\cD \xi\| \|u^{(k)}\|+o (\|u^{(k)}\|) \eeqn
and hence
\beq
\|u^{(k+1)}\|\leq \|\cD\|_\perp \|u^{(k)}\|+o (\|u^{(k)}\|).
\eeq
By induction on $k$ with $u^{(1)}$ sufficiently small, we have the desired result \eqref{41.3}.

\end{proof}

\commentout{
The expression \eqref{60.2} can be simplified as follows.
\begin{prop} If $\lamb_2^{(l)}<1, l=1,2,$ then 
\[
\|\cD\|_\perp= \max \{\|\cB_2^\top \cB_1\xi\|: {\xi\in \IR^{2n}, \xi\perp \xi_1},\|\xi\|=1\}
\]
\end{prop}

\begin{proof}
With $\lamb_2^{(l)}<1, l=1,2,$ the set $\{\xi\in \IR^{2n}: \xi\perp\xi_1, \xi\perp\xi_{2n} \}$ is 
co-dimension 2 invariant subspace under $\cB_l\cB_l^\top, l=1,2$.
\end{proof}
}

\section{Numerical experiments}\label{sec:num}

\subsection{Test images}

Let $C,B$ and $ P$ denote  the $256\times256$ non-negatively valued Cameraman, Barbara and Phantom images, respectively. 

For one-pattern simulation, we use $C$ and $P$ for test images.  For the two-pattern simulations, we use the complex-valued images,
Randomly Signed Cameraman-Barbara (RSCB) and Randomly Phased Phantom (RPP),  constructed as follows.\begin{description}
\item[RSCB]
Let $\{\beta_{R}(\bn)=\pm 1\}$ and $\{\beta_{I}(\bn)=\pm 1\}$ be i.i.d. Bernoulli random variables. Let \begin{eqnarray*}
x_{0}=\beta_R\odot C+i \beta_I \odot B. 
\end{eqnarray*}
\item[RPP]
Let $\{\phi(\bn)\}$ be i.i.d. uniform random variables over $[0,2\pi]$ and let 
\beqn
x_{0}=P\odot  e^{i\phi}. 
\eeqn
\end{description}

We  use the relative error (RE) 
\begin{eqnarray*}
\hbox{\rm RE}=\underset{\theta\in [0,2\pi)}{\min}\|x_{0}-e^{i\theta}x\|/\|x_{0}\|
\end{eqnarray*}
as the figure of merit
and the relative residual (RR) 
\begin{eqnarray*}
\hbox{\rm RR} =\|b- |A^* x|\|/\|x_{0}\|
\end{eqnarray*}
as a metric for setting the stopping rule. 

\subsection{Wirtinger Flow}

WF is a two-stage algorithm proposed by \cite{CLS2} and further improved by
\cite{truncatedWF} (the truncated version). 

The first stage is the spectral initialization (Algorithm 2). 
For the truncated spectral initialization  \eqref{truncated_version}, the parameter $\tau$ can be  optimized  by tracking and minimizing  the residual $\|b-|A^* x_k|\|$.

The second stage is
a gradient descent method  for the cost function
\beq
\label{6.1}
F_{\rm w}(x)={N\over 2} \||A^*x|^2-b^2\|^2 
\eeq
where a proper normalization is introduced to adjust for notational difference and facilitate
a direct comparison between the present set-up ($A^*$ is an isometry) and that of
\cite{CLS2}. A motivation for using \eqref{6.1} instead of  \eqref{3'} is its global differentiability. 

Below we consider these two stages separately and use the notation WF to denote primarily the second stage defined by 
the WF map
\beq\label{wf}
W(x^{(k)})&=&x^{(k)}-{s^{(k)} \over \|x^{(1)}\|^2}\nabla F_{\rm w}(x^{(k)})\\
&=&x^{(k)}-{s^{(k)} \over \|x^{(1)}\|^2} A \lt(N\lt(|A^* x^{(k)} |^2-|b|^2\rt)\odot A^* x^{(k)}\rt),\nn
\eeq
for  $k=1,2,\cdots,$ with  $s^{(k)}$ is the step size at the $k$-th iteration. Each step of WF involves twice FFT and once pixel-wise operations, comparable to the computational complexity of one iteration of PAP. 

In \cite{CLS2} (Theorem 5.1),  a basin of attraction at $x_0$ of radius $O(n^{-1/2})$ is established for $W$ 
 for a sufficiently small constant step size $s^{(k)}=s$.   
No explicit bound on $s$ is given.  As pointed out in \cite{CLS2},
 the effective step size  $s\|x^{(1)}\|^{-2}$ is inversely proportional to $\|x^{(1)}\|^2$.
 
In comparison, consider 
 the projected gradient formulation of PAP
 \beq
 \label{8.3}
 \cF(x)&= &x-\nabla F(x)\\
 &=&x-A\lt(\lt({\bf 1}-{b\over |A^*x|}\rt)\odot A^*x\rt)\nn
 \eeq
which is well-defined locally at $x_0$ and  can be extended globally by selecting
an element from the subdifferential of $F$.  

Eq. \eqref{8.3} implies a constant
step size $1$, which is significantly larger than the optimal step size for \eqref{wf}  from experiments (see below).  It is possible to improve the numerical performance of WF with 
  a heuristic dynamic step size as proposed by \cite{CLS2}, eq. (II.5), 
 \[
 s^{(k)}=\min\lt(1-e^{-k/k_0}, s_{\rm max}\rt)
 \]
 with experimentally determined $k_0, s_{\rm max}$. The performance of this ad hoc rule 
 can be sensitive to the set-up (image size, measurement scheme etc). For example, the numerical
 values $k_0=330$ and $s_{\rm max}=0.4$ suggested by \cite{CLS2} often lead
 to instability in our setting.  Since such a dynamic rule does not yet enjoy any performance guarantee, 
 we will not consider it further.

In addition, it may be worthwhile to compare  the ``weights" in $\nabla F_{\rm w}$ and $\nabla F$: \beq
\label{8.4}
N\lt(|A^* x^{(k)} |^2-|b|^2\rt)=N|A^* x^{(k)}|^2 \lt({\bf 1}-{|b|^2\over |A^*x|^2}\rt)\quad\hbox{ in $\nabla F_{\rm w}$}
\eeq
 versus
\beq
\label{8.5} \lt({\bf 1}-{b\over |A^*x|}\rt) \quad\hbox{in $\nabla F$}.
\eeq
Notice that the factor $N|A^* x^{(k)}|^2(j)$ in \eqref{8.4} is approximately $N b^2(j),\forall j,$ when $x^{(k)}\approx x_0$ while the corresponding factor in \eqref{8.5} is uniformly 1 independent of $x^{(k)}$. Like the truncated spectral initialization, the truncated Wirtinger Flow  seeks to reduce the variability of the weights in \eqref{8.4} by
introducing 3 new control parameters \cite{truncatedWF}.

\commentout{
Unfortunately, the resampled WF (\cite{CLS2}, Algorithm 2)  requires a large number of coded diffraction patterns in practice, not just in theory, and can not possibly work with two coded patterns or less. Hence instead of the resampled WF we test the original version \eqref{wf} whose convergence, however,  is proved only for Gaussian i.i.d matrices $A$ with a constant step size $s=O(1/n)$, resulting in a pessimistic rate of convergence $(1-{c\over n})$ (\cite{CLS2}, Theorem 3.3). 
}
 

\commentout{%
 In unconstrained optimization, convergence (not necessarily to the global minimum point) for gradient descent with a constant step size $s$  holds  if a Lipschitz constant $L$ of $\nabla F$ satisfies  $0<s<2/L$ (see Proposition  1.2.3 \cite{Np}).
 For the quartic cost function \eqref{6.1}, this general result only leads to local convergence for any step size.
 }


\commentout{
\begin{equation}
s^{(k)}=\min(1-e^{-k/330},0.4).
\end{equation}
}



\commentout{
\subsection{Initialization}
\label{sec:init}

Initialization is an important part of any iterative schemes. 
Here we compare the null initialization with the spectral initialization used in \cite{CLS2} and the truncated spectral initialization used in \cite{truncatedWF}. We do not
consider the resampled initialization scheme of \cite{CLS2} since it requires
a large number of coded diffraction patterns for numerical implementation. 

Let $\mathbf{1}_{c}$ be the characteristic function of the complementary index $I_c$
with $|I_c|=\gamma N$. Note that $\gamma+\sigma=1$ with $\sigma$ given by \eqref{53'}.

\begin{algorithm}
\SetKwFunction{Round}{Round}

\textbf{Random initialization:} $x_{1}=x_{\rm rand}$
\\
\textbf{Loop:}\\
\For{$k=1:k_{\textup{max}}-1$}
{
$x'_{k}\leftarrow A (\mathbf{1}_{c}\odot A^*x_k)$;\\
$x_{k+1}\leftarrow \lt[x_k^{'}\rt]_\cX/\|\lt[x_k^{'}\rt]_\cX\|$
}
{\bf Output:} {$x_{\textup{null}}=x_{k_{\textup{max}}}$.}
\caption{\textbf{The  null initialization}}
\label{null-algorithm}
\end{algorithm}
In Algorithm \ref{null-algorithm}, the default choice for $\gamma$ is the  median value $\gamma=0.5$.

\begin{algorithm}[h]
\SetKwFunction{Round}{Round}
\textbf{Random initialization:} $x_1=x_{\rm rand}$\\
\textbf{Loop:}\\
\For{$k=1:k_{\textup{max}}-1$}
{
$x_k'\leftarrow A(|b|^2\odot A^*x_k);$\\
$x_{k+1}\leftarrow \lt[x_k^{'}\rt]_\cX/\|\lt[x_k^{'}\rt]_\cX\|$;
}
{\bf Output:} $x_{\rm spec}=x_{k_{\rm max}}$.
\caption{\textbf{The spectral initialization}}
\label{spectral-algorithm}
\end{algorithm}

The key difference between the null initialization and the spectral initialization  is the different weights used in
step 4 where  the null initialization uses $\mathbf{1}_{c}$ and  the spectral vector
method uses $|b|^2$ (Algorithm 2). 
The truncated spectral initialization uses  a still  different weighting 
\begin{equation}\label{truncated_version}
x_{\textup{t-spec}}=\textup{arg}\underset{\|x\|=1}{\textup{max}}
\| A
\left(\mathbf{1}_\tau\odot |b|^{2}\odot A^*x \right)\|
\end{equation}
where 
$\mathbf{1}_\tau$ is the characteristic function of 
the set
\[
\{i: |A^* x(i)|\leq \tau {\|b\|}\}
\]
with an adjustable parameter $\tau$. 
As shown below the choice of weight  significantly
affects  the quality of initialization, with the null initialization as  the best performer
consistent with Remark \ref{rmk5.2}. 

Both $\gamma$ of Algorithm 1 and $\tau$ of \eqref{truncated_version} can be 
optimized  by tracking and minimizing  the RR.

}

\begin{figure}[t]
\centering
\subfigure[][$x_{\rm spec}$]{\label{fig:spectral_initial3}\includegraphics[scale=0.7]{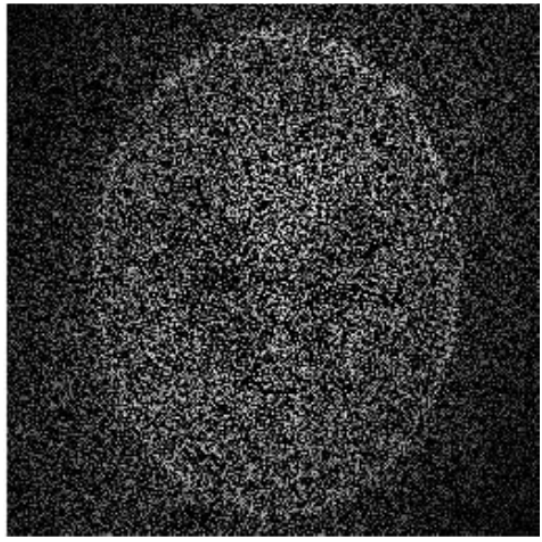}}
\subfigure[][$x_{\textup{t-spec}}$ ($\tau^2=4.6$)]{\label{fig:Tspectral_initial3}\includegraphics[scale=0.7]{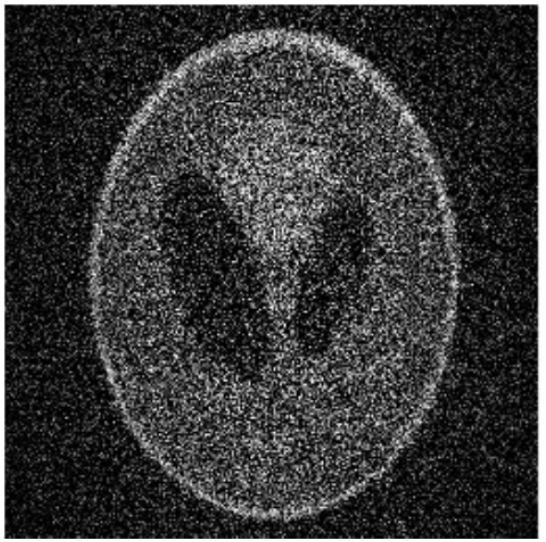}}
\subfigure[][$x_{\textup{null}}$ ($\gamma=0.5$) ]{\label{fig:null_initial3_neutral}\includegraphics[scale=0.7]{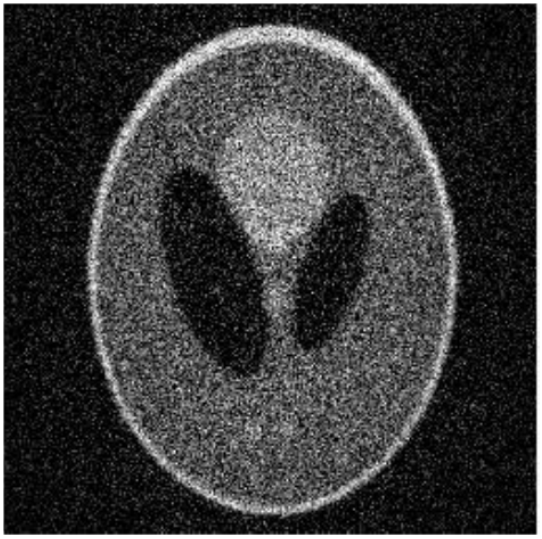}}
\subfigure[][$x_{\textup{null}}$ ($\gamma=0.74$)]{\label{fig:null_initial3}\includegraphics[scale=0.7]{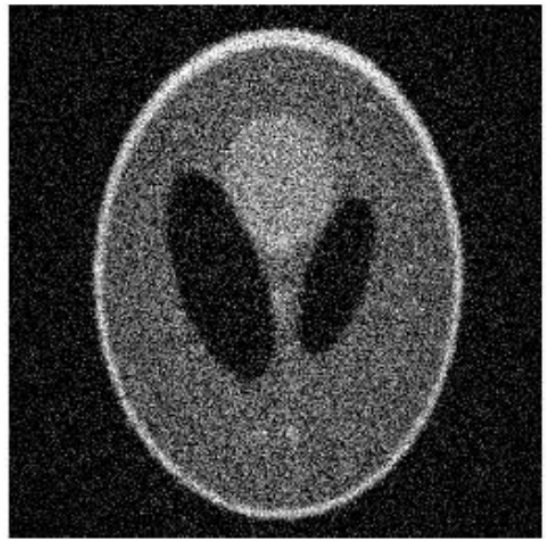}}
\subfigure[][$x_{\rm spec}$]{\label{fig:L1spectral_initial}\includegraphics[scale=0.7]{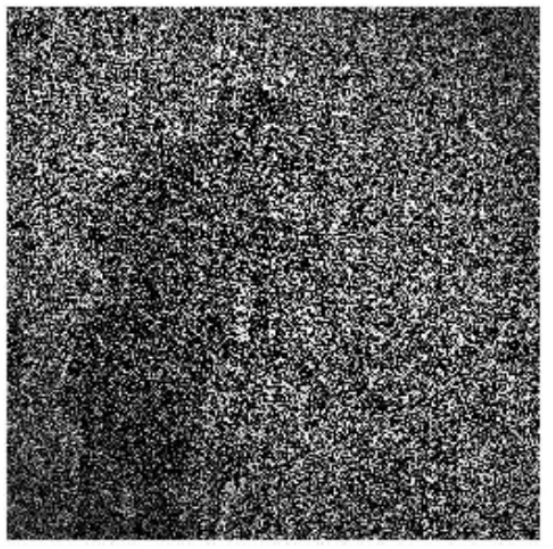}}
\subfigure[][$x_{\textup{t-spec}}$ ($\tau^{2}=4.1$)]{\label{fig:L1Tspectral_initial}\includegraphics[scale=0.7]{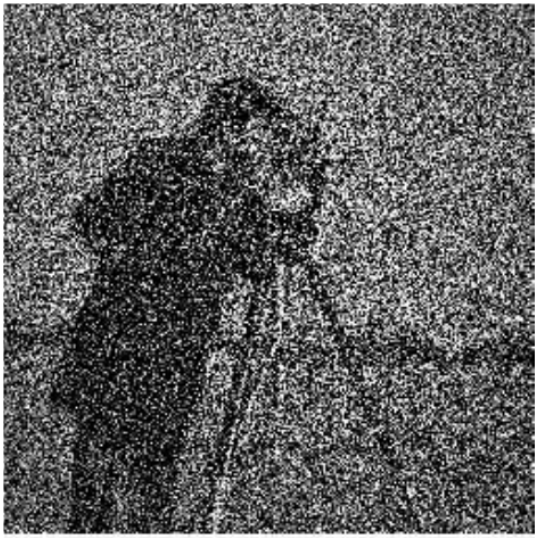}}
\subfigure[][$x_{\textup{null}}$ ($\gamma=0.5$)]{\label{fig:L1null_initial_neutral}\includegraphics[scale=0.7]{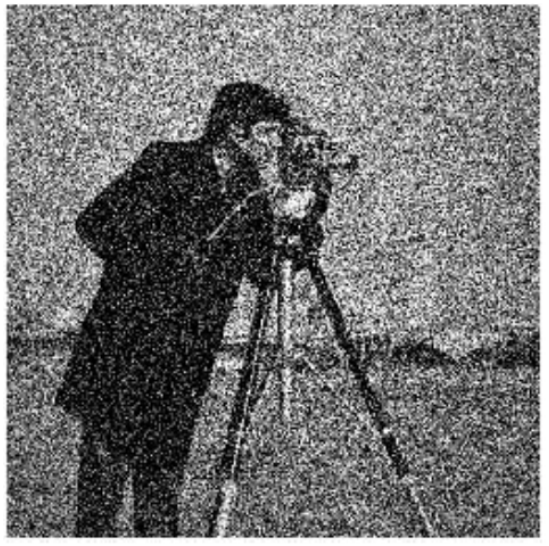}}
\subfigure[][$x_{\textup{null}}$ ($\gamma=0.7$)]{\label{fig:L1null_initial}\includegraphics[scale=0.7]{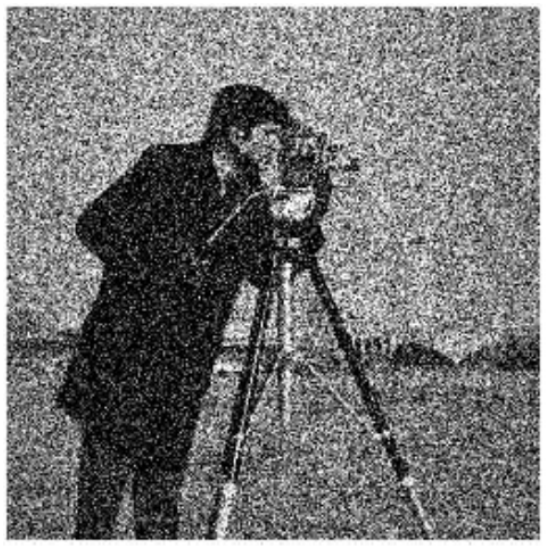}}
\caption{Initialization with one pattern  of the Phantom ((a) \hbox{\rm RE}$(x_{\rm spec})=0.9604$,
(b) \hbox{\rm RE}$(x_{\textup{t-spec}})=0.7646$, (c) \hbox{\rm RE}$(x_{\textup{null}})=0.5119$,
(d) \hbox{\rm RE}$(x_{\textup{null}})=0.4592$) and  the Cameraman ((e) \hbox{\rm RE}$(x_{\rm spec})=0.8503$, (f) \hbox{\rm RE}$(x_{\textup{t-spec}})=0.7118$, (g) \hbox{\rm RE}$(x_{\textup{null}})=0.4820$, (h) \hbox{\rm RE}$(x_{\textup{null}})=0.4423$).}
\label{fig:initials3}\label{fig:initials}
\end{figure}

\begin{figure}[h]
\centering
\subfigure[][Cameraman]{\includegraphics[scale=0.96]{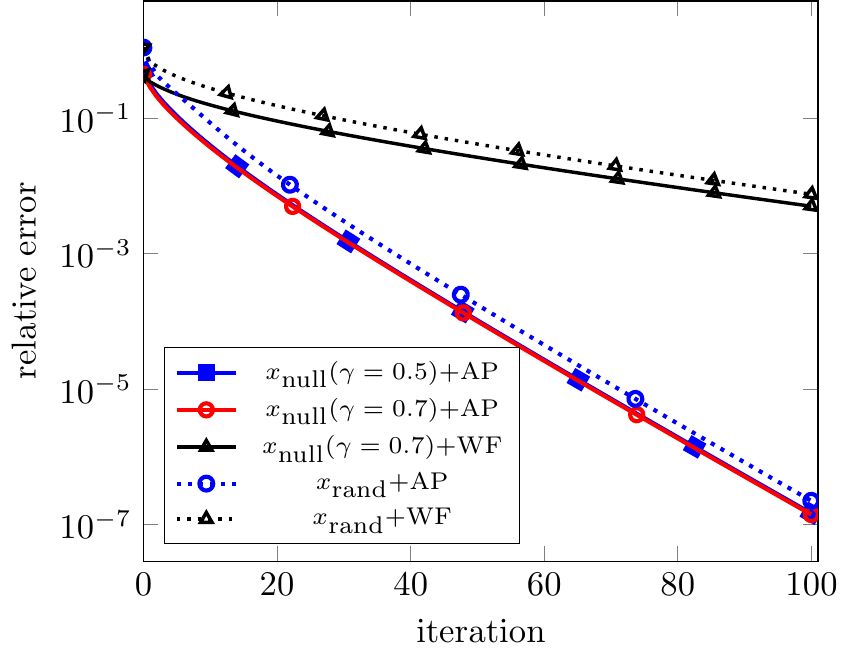}}
\subfigure[][Phantom]{\includegraphics[scale=0.96]{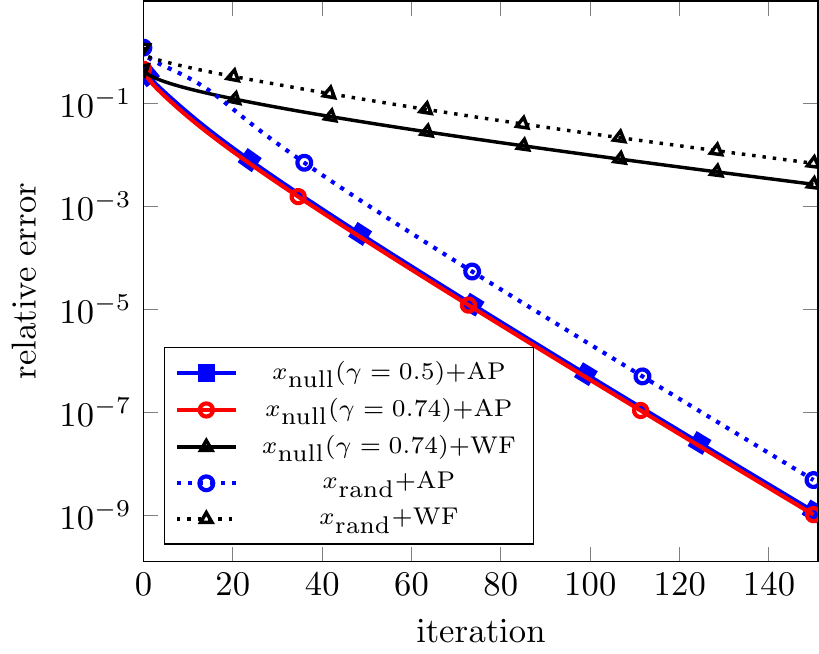}}
\caption{RE versus iteration in the one-pattern case with the (a) Cameraman and (b) Phantom. WF is tested with the optimized  step size $s=$ 0.2.}
\label{fig:positive_phantom_case}\label{fig5}\label{fig4}\end{figure}

\begin{figure}[tbp]
\centering
\subfigure[][$|\textup{Re}(x_{\textup{t-spec}})|$ ($\tau^2=5$)]{\label{fig:Tspectral_initial}\includegraphics[scale=0.7]{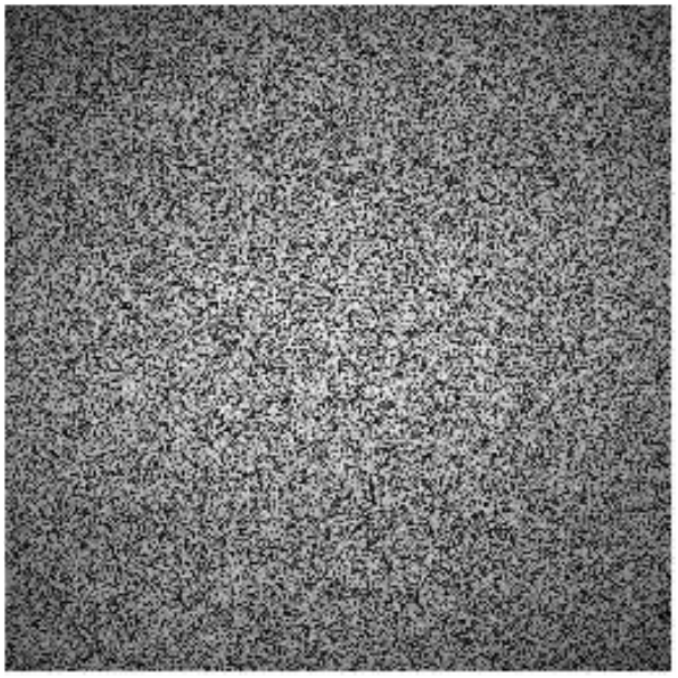}}
\subfigure[][$|\textup{Re}(x_{\textup{null}})|$ ($\gamma=0.5$)]{\label{fig:null_initial_neutral_real}\includegraphics[scale=0.7]{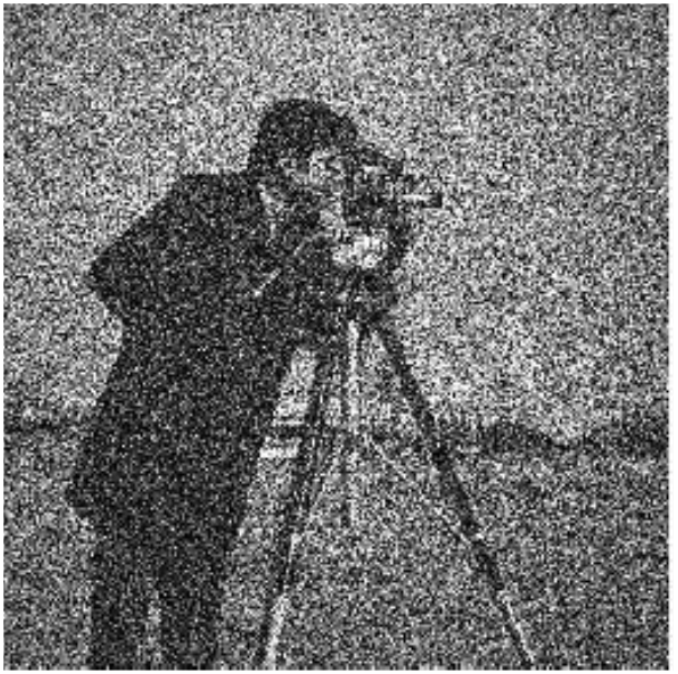}}
\subfigure[][$|\textup{Re}(x_{\textup{null}})|$ ($\gamma=0.63$)]{\label{fig:null_initial_real}\includegraphics[scale=0.7]{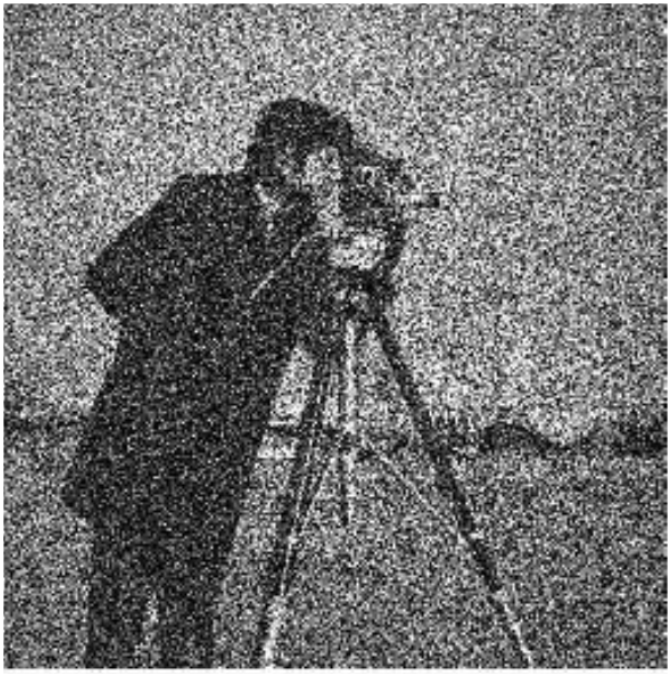}}
\\
\subfigure[][$|\textup{Im}(x_{\textup{t-spec}})|$ ($\tau^2=5$)]{\label{fig:Tspectral_initial'}\includegraphics[scale=0.7]{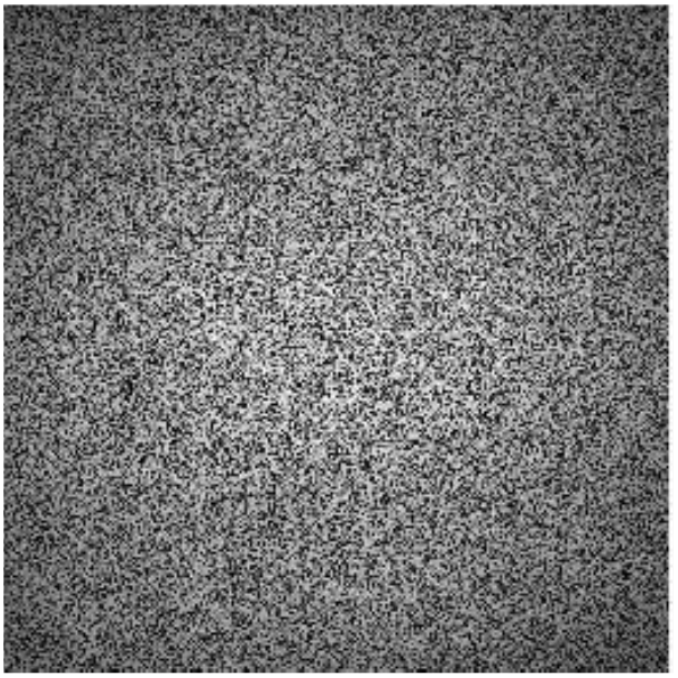}}
\subfigure[][$|\textup{Im}(x_{\textup{null}})|$ ($\gamma=0.5$)]{\label{fig:null_initial_neutral_imag}\includegraphics[scale=0.7]{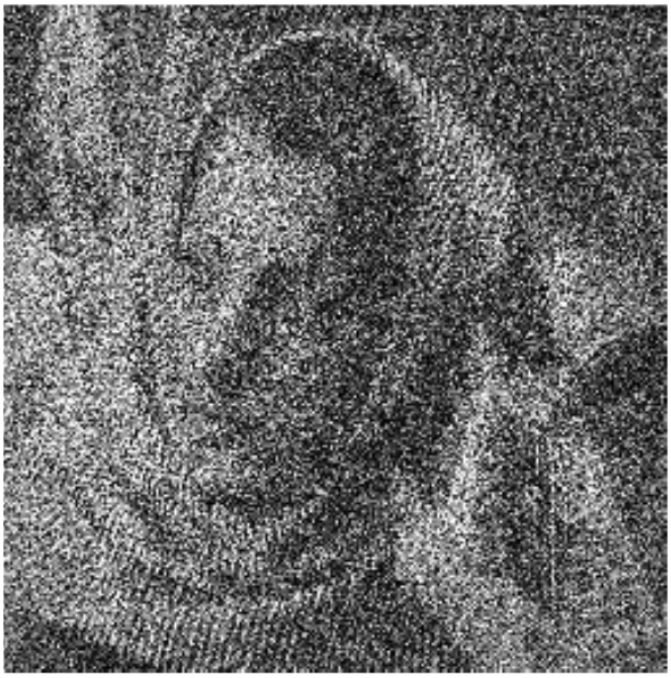}}
\subfigure[][$|\textup{Im}(x_{\textup{null}})|$ ($\gamma=0.63$)]{\label{fig:null_initial_imag}\includegraphics[scale=0.7]{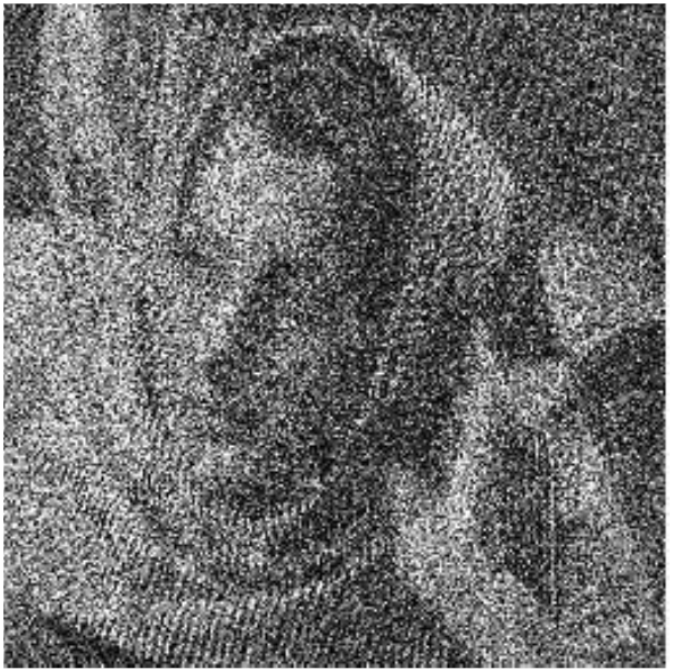}}
\subfigure[][$|x_{\textup{t-spec}}|$ ($\tau^2=5$)]{\label{fig:Tspectral_abs_phantom}\includegraphics[scale=0.7]{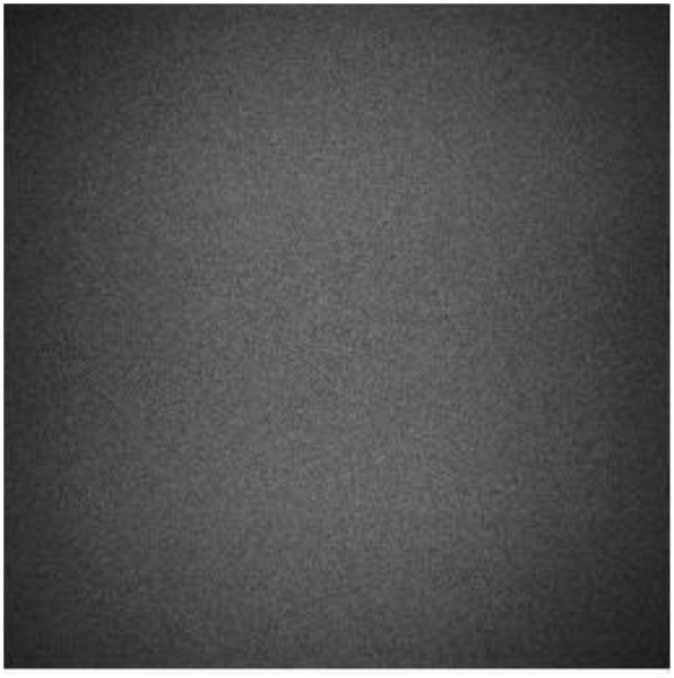}}
\subfigure[][$|x_{\textup{null}}|$ ($\gamma=0.5$)]{\label{fig:null_initial_neutral_abs_phantom}\includegraphics[scale=0.7]{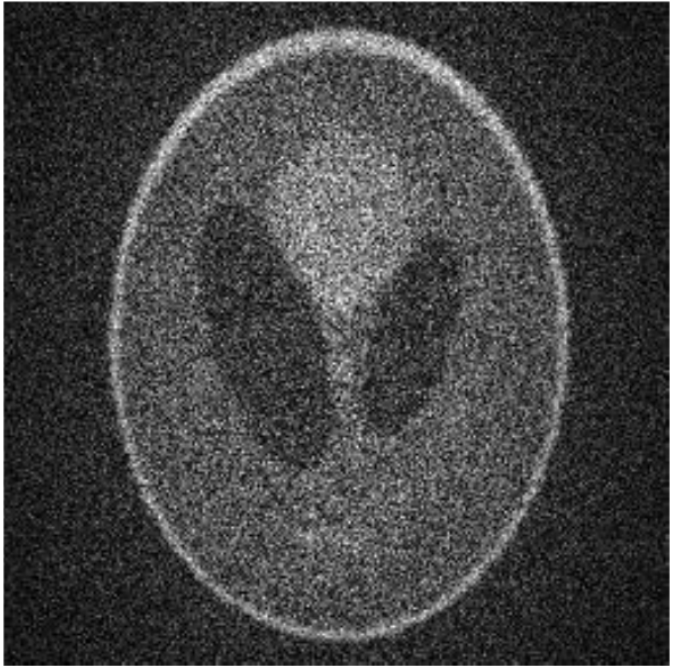}}
\subfigure[][$|x_{\textup{null}}|$ ($\gamma=0.6$)]{\label{fig:null_initial_abs_phantom}\includegraphics[scale=0.7]{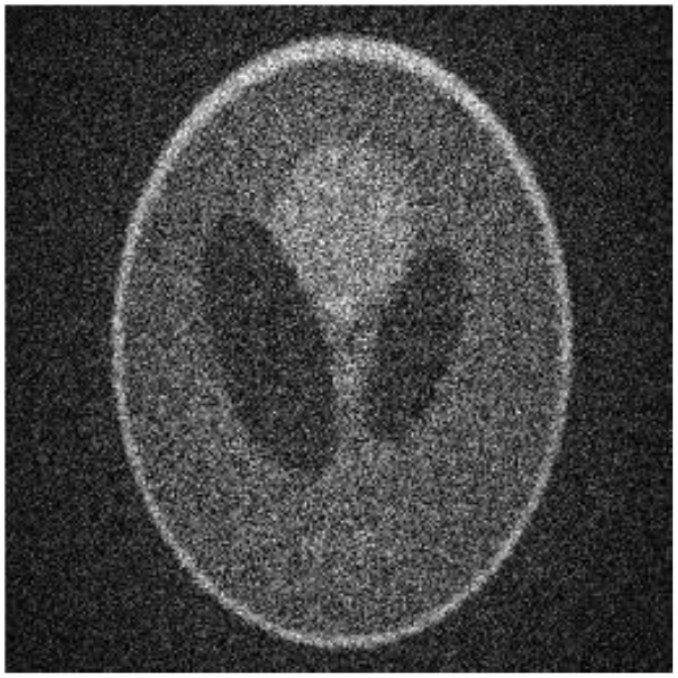}}
\caption{ Initialization with two patterns for RSCB ((a)(d) \hbox{\rm RE}$(x_{\textup{t-spec}})=1.3954$, (b)(e) \hbox{\rm RE}$(x_{\textup{null}})=0.5736$,  (c)(f) \hbox{\rm RE}$(x_{\textup{null}})=0.5416$) and  RPP ((g)\hbox{\rm RE}$(x_{\textup{t-spec}})=1.3978$,
(h) \hbox{\rm RE}$(x_{\textup{null}})=0.7399$, (i) \hbox{\rm RE}$(x_{\textup{null}})=0.7153$)}
\label{fig:initials_2masks_phantom}\label{fig:initials_2masks}
\end{figure}

\subsection{One-pattern experiments}

 Fig. \ref{fig:initials}  shows that  the null vector $x_{\textup{null}}$ is more accurate than the spectral vector $x_{\rm spec}$ and the truncated spectral vector $x_{\textup{t-spec}}$ in approximating the true images.
For the Cameraman (resp. the Phantom) 
  $\hbox{\rm RR}(x_{\textup{null}})$
can be minimized by setting $\gamma\approx0.70$
(resp. $\gamma\approx0.74$).  The optimal parameter $\tau^2$ for $x_{\rm t-spec}$ in (\ref{truncated_version}) is about $4.1$ (resp. $4.6$).  

Next we compare the performances of PAP and WF \cite{CLS2}  with $x_{\textup{null}}$ as well as  the random initialization $x_{\textup{rand}}$. Each pixel of $x_{\textup{rand}}$ is independently sampled from the uniform distribution over $[0,1]$.

To account for the real/positivity constraint, we modify \eqref{wf} as
\beq\label{wf1}
W(x^{(k)})&=&\lt[x^{(k)}-{s^{(k)} \over \|x^{(1)}\|^2}\nabla F_{\rm w}(x^{(k)})\rt]_\cX,\quad \cX=\IR^n,~\IR^n_+.
\eeq

As shown in Fig. \ref{fig4}, the convergence of both PAP and WF is faster with $x_{\textup{null}}$ than  $x_{\textup{rand}}$. In all cases, PAP converges  faster than WF. 

Also, the median value  $\gamma=0.5$ for initialization  is as good as
the optimal value. The convergence of PAP with random initial condition $x_{\rm rand}$  suggests  {\em global}
convergence to the true object  in the one-pattern case with the positivity constraint. 

\subsection{Two-pattern experiments }\label{sec:WF_ER}
\begin{figure}[htbp]
\centering
\subfigure[][RSCB]{\label{fig:general_case_ERWF_null}\includegraphics[scale=0.90]{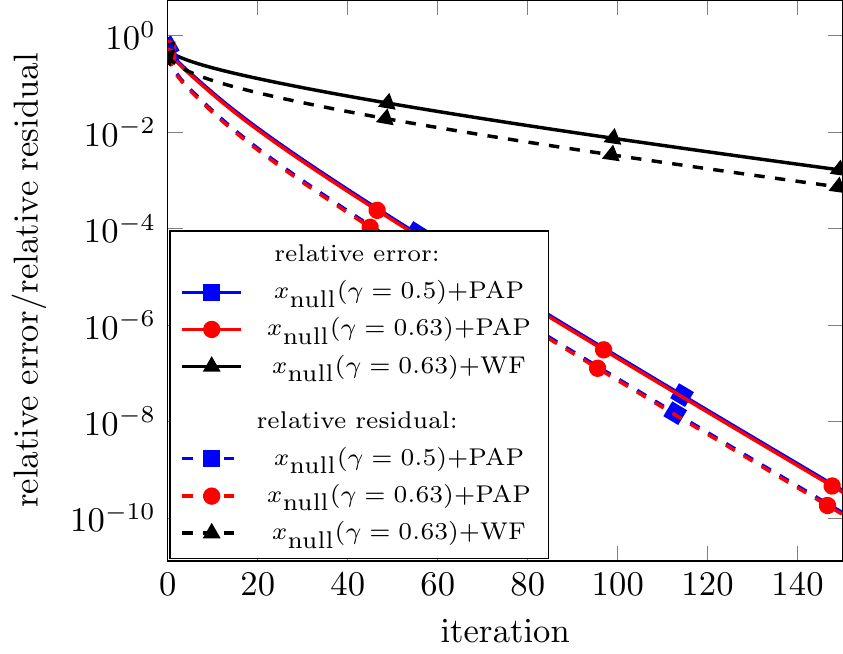}}
\subfigure[][RPP]{\label{fig:general_case_ERWF_null_phantom}\includegraphics[scale=0.90]{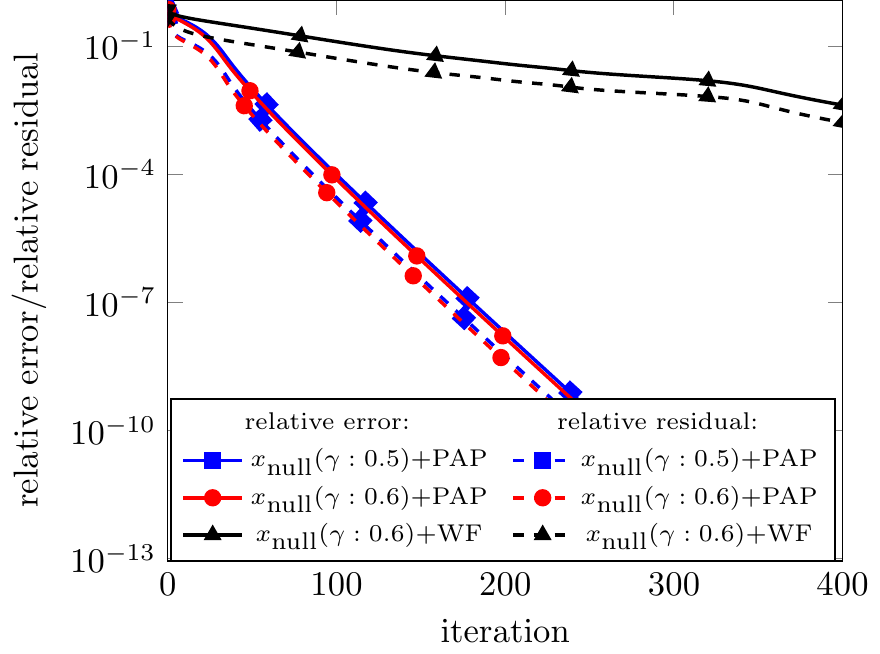}
}
\caption{RE and RR versus iteration for PAP and WF with two patterns. WF is tested with the optimized  step size (a) $s=0.2$  and (b) $s=0.15$.}
\label{fig:general_case_phantom}\label{fig:general_case}
\end{figure}

We use the complex images, RSCB and RPP,  for the two-pattern simulations. 

 Fig. \ref{fig:initials_2masks} shows  that
$x_{\rm null}$  is more accurate than the $x_{\rm spec}$ and $x_{\rm t-spec}$ in
approximating $x_0$.  The  difference in RE  between  the initializations with the median value and
the optimal values is less than $3\%$. 

Fig. \ref{fig:general_case} shows that PAP outperforms WF, both with
 the null initialization. 
 
 As Fig. \ref{fig:serial_parallel_ER} shows, SAP converges much faster than PAP
and takes about half the number of iterations to converge to the object. 
Different samples correspond to different realizations of random masks, showing robustness with respect to the ensemble of random masks.
In  terms of the rate of convergence, 
SAP with the null initialization outperforms  the Fourier-domain Douglas-Rachford algorithm \cite{DR-phasing}. 
 
 \begin{figure}[tbp]
\centering
\subfigure[][RSCB ]{\includegraphics[scale=0.9]{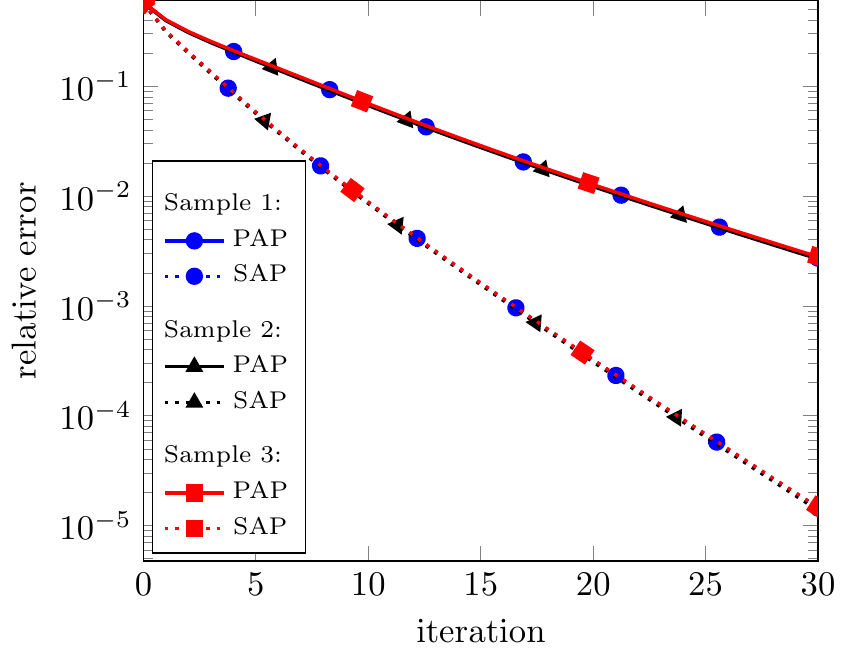}}
\hspace{0.0cm}
\subfigure[][ RPP ]{\includegraphics[scale=0.9]{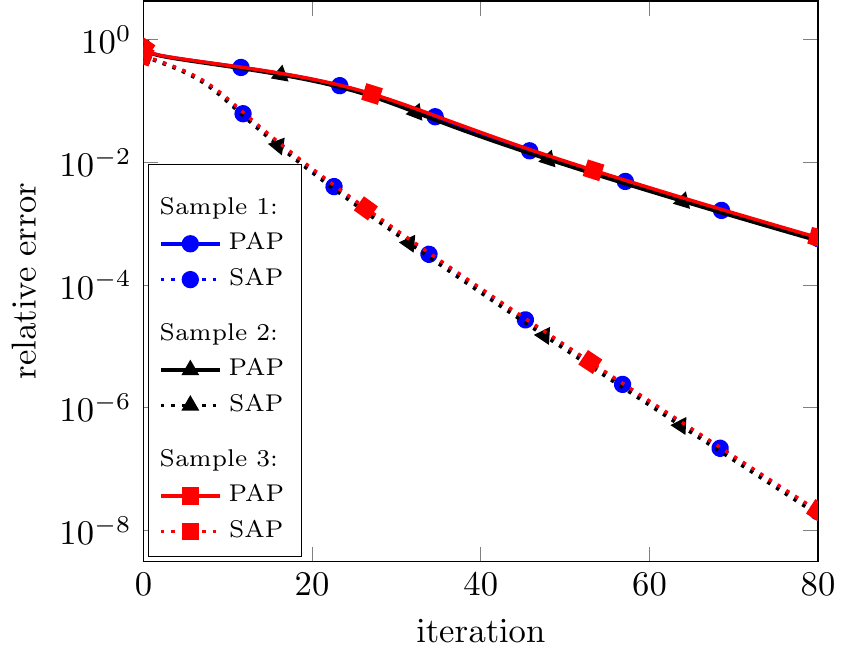}}
\caption{RE versus iteration of PAP and SAP in the two-pattern case  ($\gamma=0.5$).}
\label{fig:serial_parallel_ER}
\end{figure}

Fig. \ref{fig0} shows the RE versus iteration for the (a) one-pattern and (b) two-pattern cases.  The dotted lines represent  the geometric series $\{ \tilde \lambda_2^{2k}\}_{k=1}^{200}$,  $\{ \lambda_2^{2k}\}_{k=1}^{200}$ and $\|\cD\|_\perp^k$ (the pink line in (a) and the red and the blue lines in (b)), which track well the actual iterates  (the black-solid curve in (a) and the blue- and the red-solid curves in (b)), consistent with the predictions of Theorems \ref{thm1},  \ref{thm2} and \ref{thm3}. In particular, SAP has a better rate of convergence than PAP (0.7946 versus 0.9086).

\begin{figure}[t]
\centering
\subfigure[][One pattern: $\tilde \lambda^2_2=0.9084$]{\includegraphics[scale=.45]{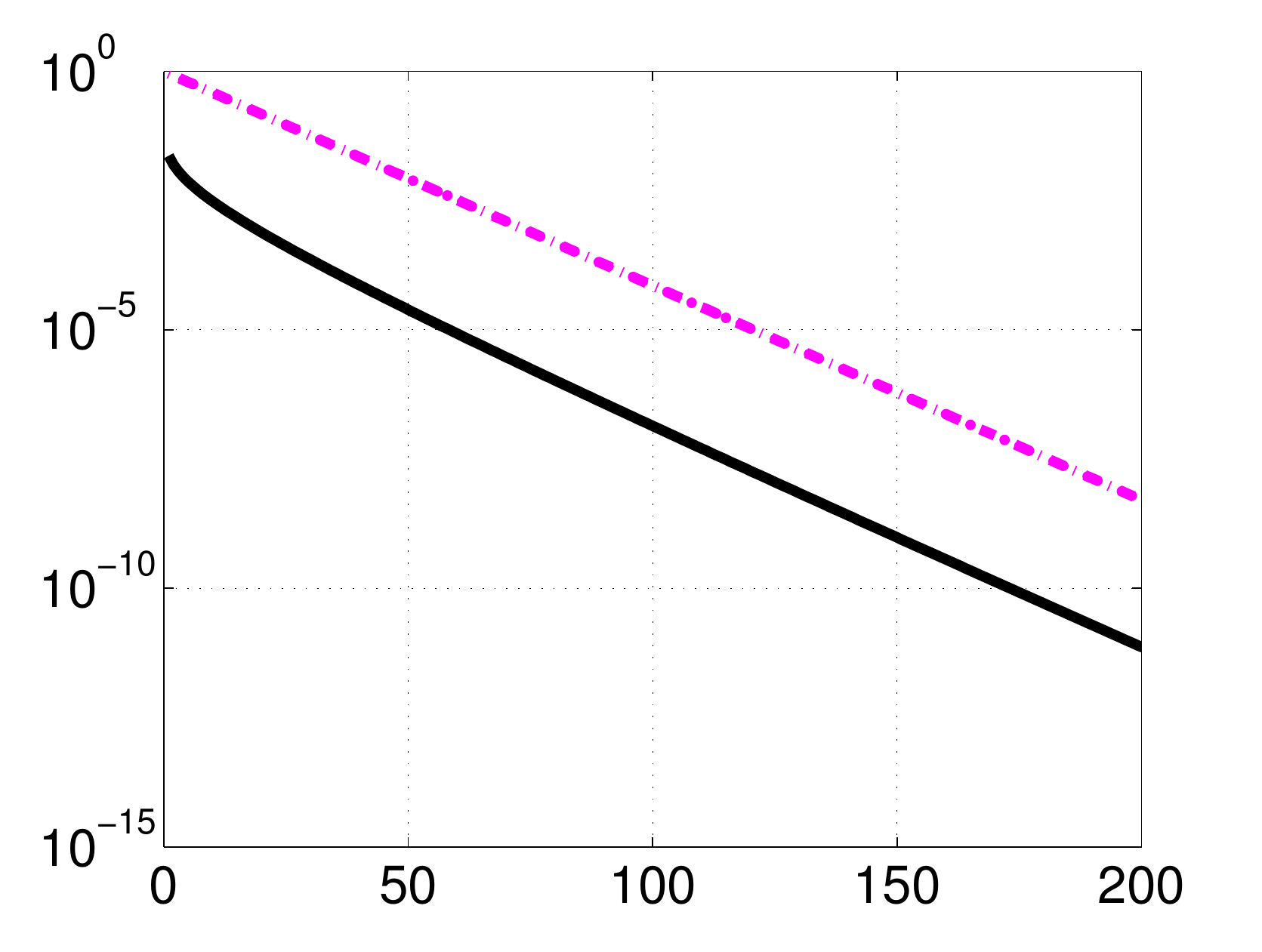}}
\subfigure[][SAP $\|\cD\|_\perp=0.7946$; PAP $\lambda^2_2=0.9086$]{
\includegraphics[width=7.7cm,height=5.7cm]{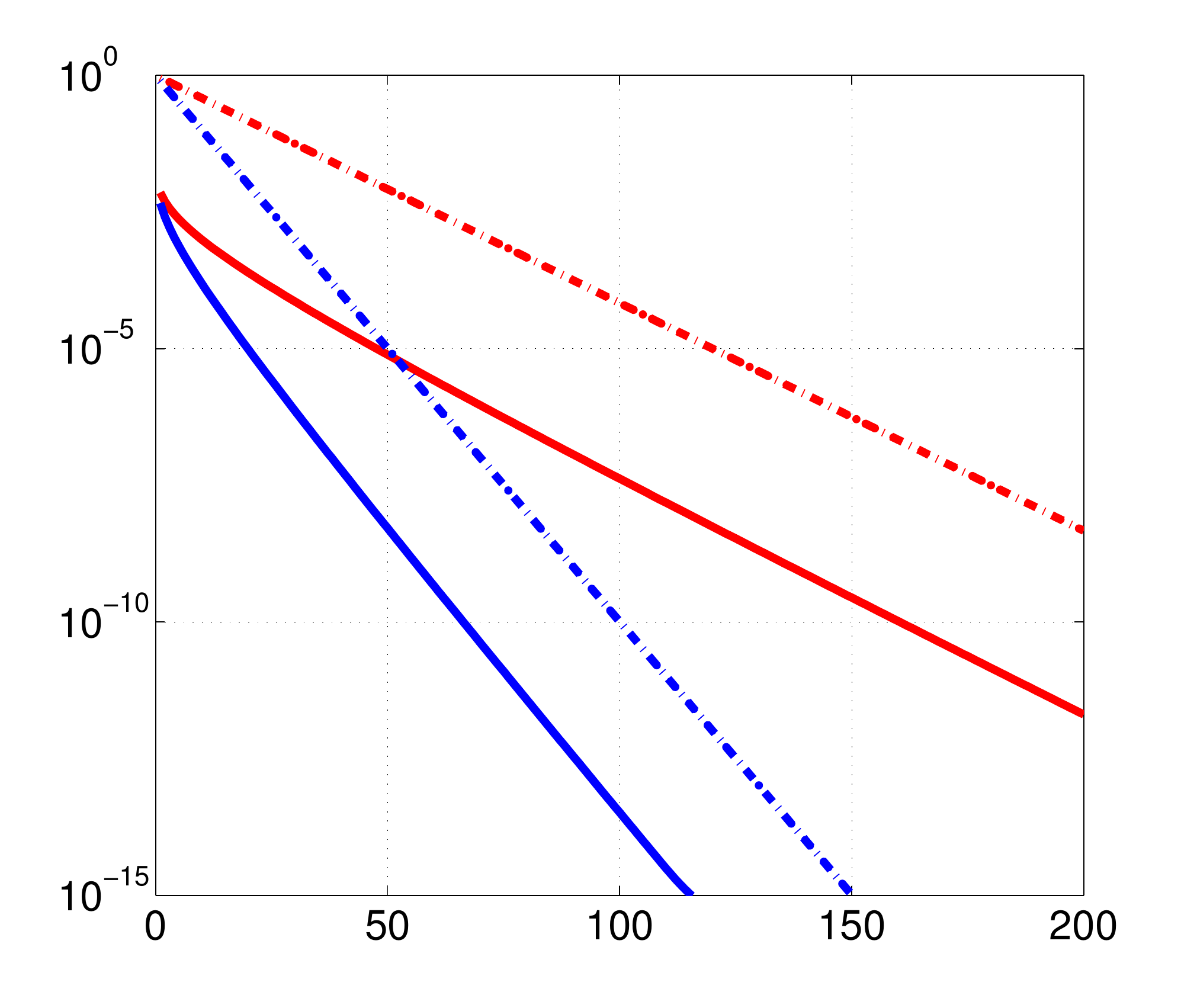}
}\caption{RE on the log scale versus iteration with (a)  one pattern and (b) two patterns
(PAP in red, SAP in blue). The solid curves are the AP iterates and the dotted lines are
the geometric series predicted by the theory.  }
\label{fig0}
\end{figure}

\subsection{Oversampling ratio}\label{sec:or}
Phase  retrieval with just one coded diffraction pattern without the real/positivity constraint
has many solutions \cite{unique} and as a result AP with the null initialization does
not perform well numerically.

What would happen if we measure two coded diffraction patterns each with
fewer samples? 

The amount of data in each coded diffraction pattern is measured by the oversampling ratio 
\begin{eqnarray*}
\rho=\frac{\textup{Number of data in {\bf each} coded diffraction pattern
}}{\textup{Number of image pixels}},
\end{eqnarray*}
which is approximately  4 in the standard oversampling. 

\commentout{
\begin{figure}[t]
\centering
\subfigure[][Cameraman]{
\includegraphics[scale=0.9]{L1_zeropadding_cameraman_combine_eps}
}
\subfigure[][Phantom]{
\includegraphics[scale=0.9]{L1_zeropadding_phantom_combine_eps}
}
\caption{RE versus iteration for the one-pattern case with various $\rho$'s and $\gamma=0.5$. }
\label{fig:effect_rho_positivephantom_caseWFER}\label{fig9}
\end{figure}
}

\begin{figure}[tbp]
\centering
\subfigure[][RSCB ($\rho=1.65$)]{\includegraphics[scale=0.9]{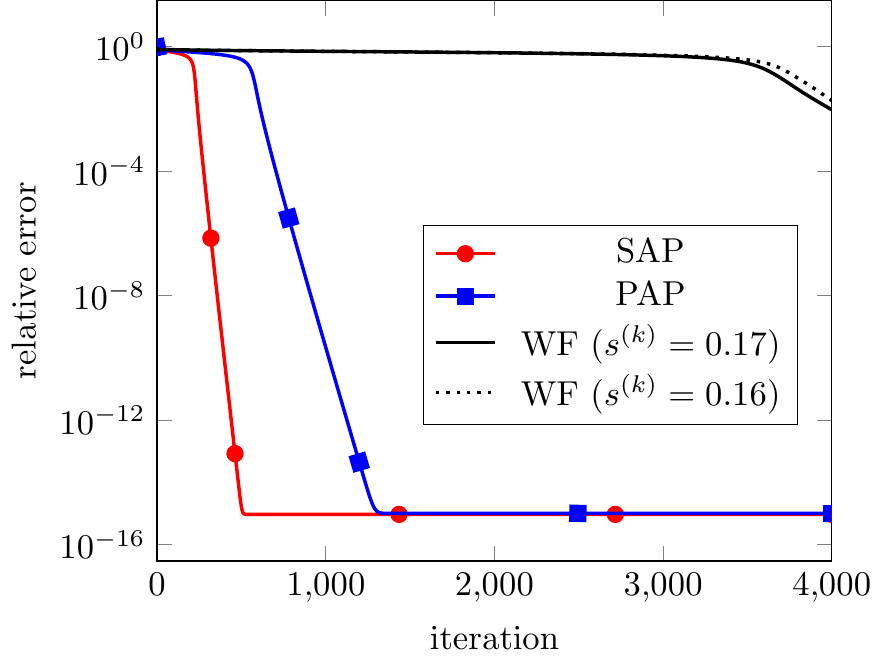}}
\subfigure[][RPP ($\rho=1.96$)]{\includegraphics[scale=0.9]{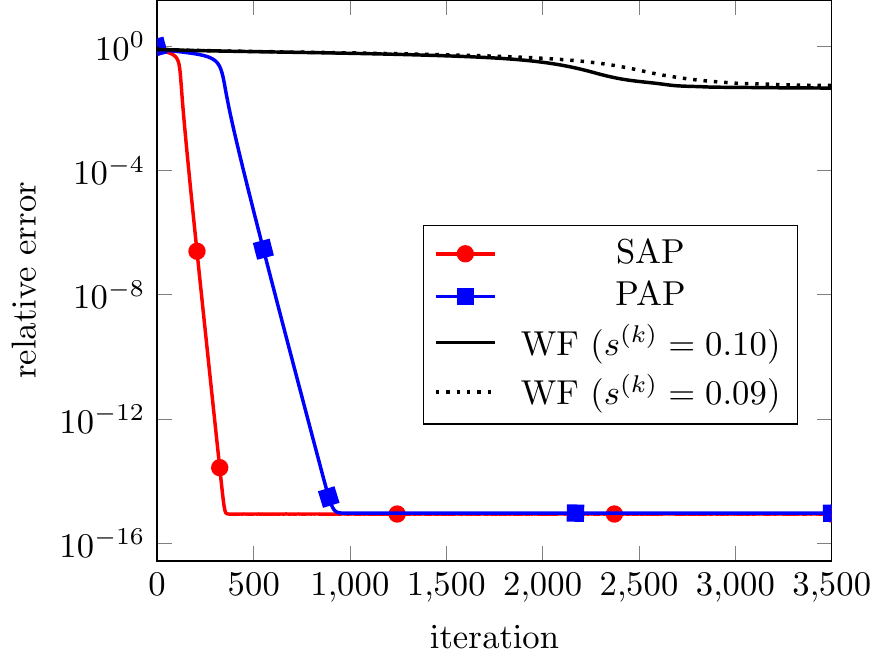}}
\caption{RE vs. iteration by SAP, PAP and WF with two patterns and the null initialization ($\gamma=0.38$ and $0.4$ for RSCB and RPP, respectively).
The optimal step size for WF is $s=0.17$ and $0.10$ for RSCB and RPP, respectively.
}
\label{fig:ERPB}\label{fig10}\end{figure}


For the two-pattern results in Fig. \ref{fig10}, we use $\rho=1.65, 1.96$ (respectively for RSCB and RPP) and hence $N\approx 3.3 n, 3.92 n$  (respectively for RSCB and RPP). 
For $n=256\times 256$, $3.3n\approx 216269, 3.92n\approx 256901$ are both
 significantly less than $(2\sqrt{n}-1)^2=261121$, the number of data in a coded diffraction pattern with the standard oversampling.
 
As expected, convergence is slowed down for both methods (much less so for SAP) as the oversampling ratio
decreases. Nevertheless, both SAP and PAP converge rapidly to the true solution, reaching machine precision, within 500 and 1200 iterations, while WF fails to converge within 4000 steps for RSCB and stagnates after 3000 iterations for RPP. 
 The optimal constant step size for WF is $s = 0.10$ and $0.17$ for RSCB and RPP, respectively. And  if we set $s= 0.11$ and $0.18$ respectively, then the relative error would blow up for both images.  On the other hand, a smaller step size results in even worse performance.


\subsection{Noise stability}
\begin{figure}[t]
\centering
\subfigure[][One pattern]{\includegraphics[scale=0.9]{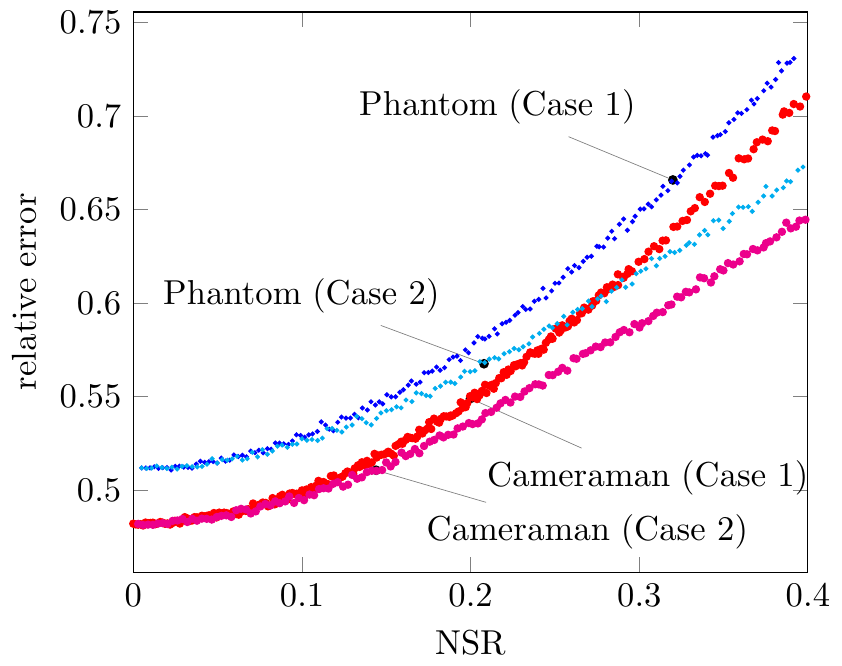}}
\subfigure[][Two patterns]{\label{fig:L2_noiseeffect_nullvector}\includegraphics[scale=0.9]{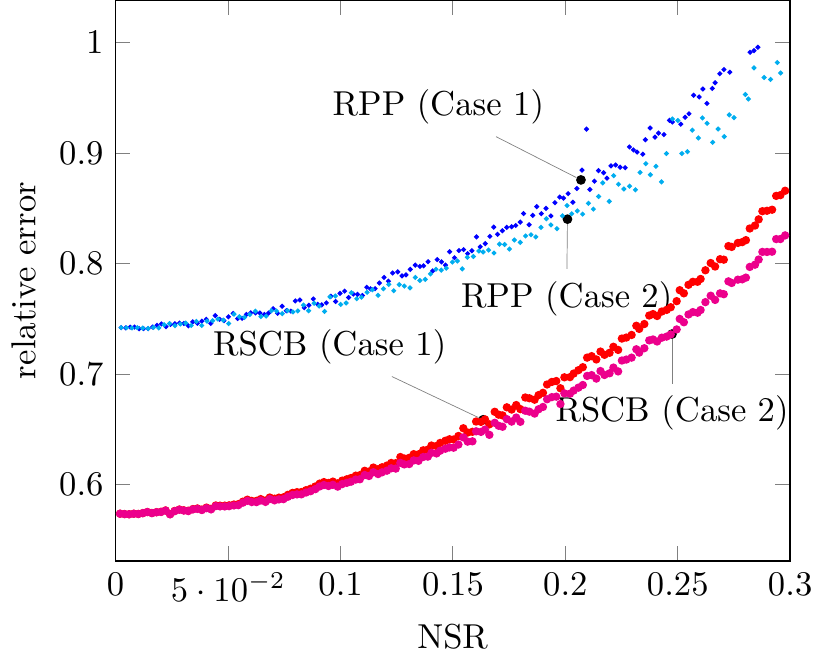}}
\caption{RE versus NSR of  the null initialization ($\gamma=0.5$)}
\label{fig:noise0}
\end{figure}
\begin{figure}[t]
\centering
\subfigure[][Cameraman]{\includegraphics[scale=0.9]{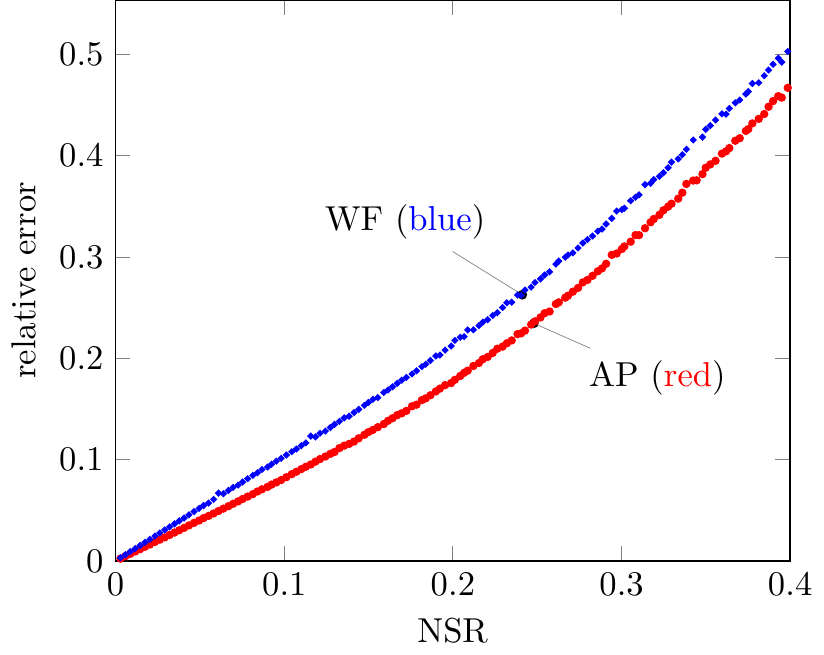}}
\subfigure[][Phantom]{\includegraphics[scale=0.9]{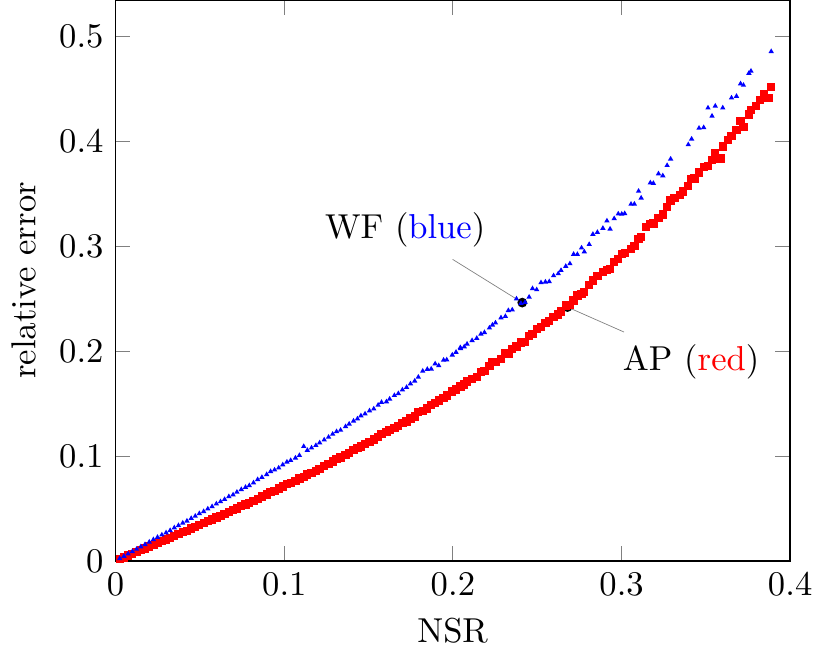}}
\subfigure[][RSCB]{\label{fig:L2_noiseeffect_RSCB}\includegraphics[scale=0.9]{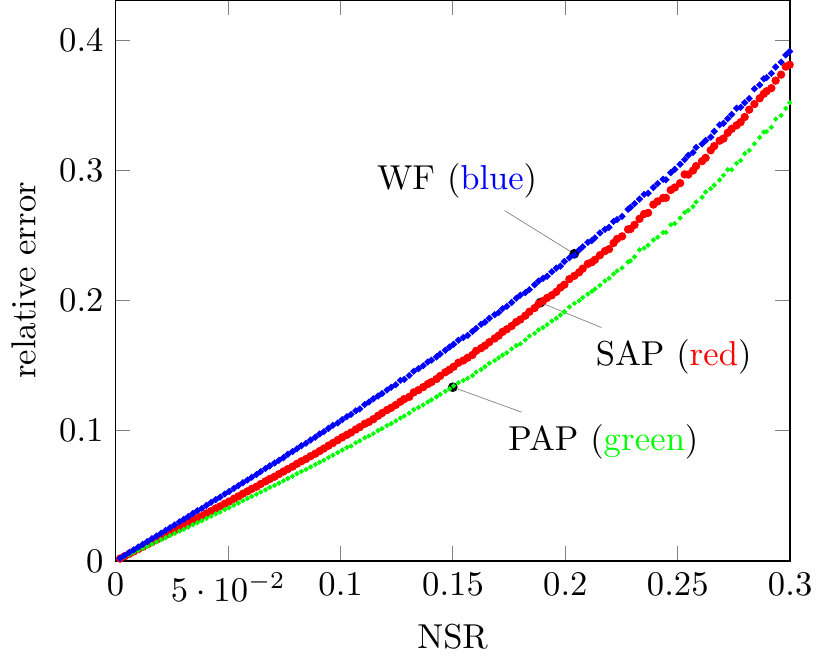}}
\subfigure[][RPP]{\label{fig:L2_noiseeffect_RPP}\includegraphics[scale=0.9]{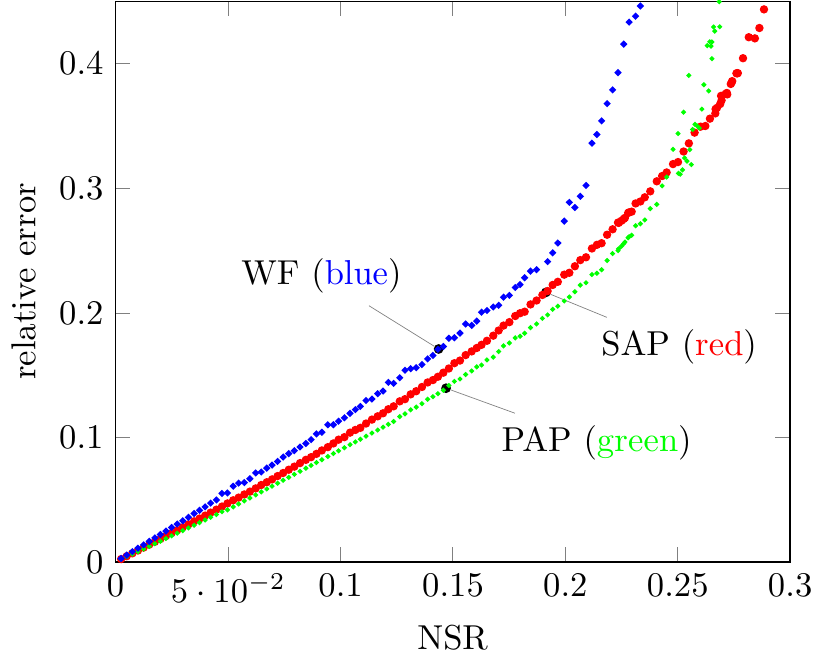}}
\caption{RE versus NSR with one (top row, 500 iterations) and two  (bottom row, 1000 iterations) patterns.}
\label{fig:noise}
\end{figure}

We test the performance of AP and WF with the Gaussian noise model
where the noisy data is generated by 
\[
b_{\rm noisy}=|A^{*}x_{0}+ \hbox{complex\ Gaussian\ noise} |.
\]
The noise  is measured by the Noise-to-Signal Ratio (NSR)
\[
\textup{NSR}=\frac{\|b_{\rm noisy}-|A^* x_0|\|_2}{\|A^* x_0\|_2}.
\]

As  pointed out in Section \ref{sec:spectral},  
since  the null initialization depends only on the choice of the index set $I$  and does not depend explicitly on $b$, the method  is more noise-tolerant than other initialization methods. 

{ Let $\hat x_{\rm null}$ be the {\em unit} leading singular vector of $A_{I_c}$, cf. \eqref{2.4}. In order to compare the effect of normalization, we  normalize the null vector  in two different ways 
\beq
\hbox{\rm Case 1.}\quad x_{\rm null}&=&\alpha{\|b_{\rm noisy}\|}\cdot\hat x_{\rm dual}\\
\hbox{\rm Case 2.}\quad x_{\rm null}&=&\alpha {\|x_0\|}\cdot\hat x_{\rm dual}
\eeq
and then compute their respective relative errors versus NSR.
As shown in Fig. \ref{fig:noise0}, the  slope of RE versus  NSR is less than 1 in  all cases. Remarkably, the slope is much smaller than 1 for small NSR when the
performance curves are strictly convex and independent of the way of normalization.  
For large SNR ($\ge 20\%$), however, the proper normalization with $\|x_0\|$ (Case 2) can significantly reduce the error.
The difference between the initialization errors of RPP and RSCB would disappear by and large  after the AP iteration converges, see Fig. \ref{fig:noise}. }

Fig. \ref{fig:noise} shows the REs of AP and WF with the null initialization after 500 iterations for the one-pattern case
and 1000 iterations for the two-pattern case.   Clearly, AP consistently achieves a smaller error than WF, with a noise amplification factor  slightly above 1.
For RPP, WF, PAP and SAP fail to converge in 1000 steps beyond $20\%, 25\%$ and $28\%$ NSR, respectively, hence the scattered data points. Increasing the maximum number of iterations can bring the upward ``tails" of the curves back to
roughly straightlines as in other plots. 

As in Fig. \ref{fig:noise0}, if $\|x_0\|$ is known explicitly, we can apply AP with the normalized
noisy data
\[
\hat b_{\rm noisy}=b_{\rm noisy} {\|x_0\|\over \|b_{\rm noisy}\|}
\]
and improve the performance in Fig. \ref{fig:noise}. And the improvement is particularly
significant for larger NSR. For simplicity of presentation, the results are omitted  here. 

{ 
\section{Conclusion and discussion}
Under the uniqueness framework of \cite{unique} (reviewed in Section \ref{sec:not}), we have proved local geometric convergence for AP of various forms and characterized the convergence rate in terms of a spectral gap. To our knowledge, this is the only such result besides \cite{DR-phasing} for phase retrieval with
one or two coded diffraction patterns. Other literature either demands a large number of coded diffraction patterns \cite{CLS1, CLS2} or asserts sublinear convergence \cite{Noll}. More importantly, we have proposed and proved  the null initialization to be an effective initialization method with performance guarantee comparable to, and numerical performance superior to, the spectral initialization and its truncated version \cite{CLS2, truncatedWF}. In practice AP with the null initialization is a globally convergent algorithm for phase retrieval with one or two coded diffraction patterns. 

Of course, a positive spectral gap does not necessarily lead to a significantly sized
basin of attraction for the true object. As mentioned above AP with
just {\bf one}  coded diffraction pattern but {\bf without} any object constraint  still has a positive spectral gap and converges locally to the true object. However,  AP with the null initialization  does not perform well numerically (not shown). 
This is likely because the corresponding phase retrieval  loses uniqueness and has many solutions \cite{unique}. On the other hand, AP with one coded diffraction pattern under the real or positivity constraint 
converges globally with randomly selected initial guess (Fig. \ref{fig4}) because
the uniqueness of solution is restored with the object constraint.  

This observation points to the importance of the design of measurement scheme besides the choice of algorithm (AP versus WF, e.g.). Results that do not take the measurement scheme into account (e.g. \cite{Noll}) are likely to be sub-optimal in theory and practice. 

A reasonable question is, How much can the measurement scheme be relaxed from that of \cite{unique}? Fig. \ref{fig10} gives a tentative answer to one aspect of the question: the number of measurement data may be reduced by as much as half
and still maintains a good numerical performance. Another aspect of the question is about the type of masks to be used in measurement: Indeed, besides the fine-grained (independently distributed) masks discussed in Section \ref{sec:not}, the coarse-grained
(correlated) masks can have a good numerical performance as well (see \cite{rpi,pum}).

A shortcoming of the present work is that we are unable to provide a useful estimate
for the size of the basin of attraction for AP; our current estimate is overly pessimistic (not shown).  Another is that we are unable to give an error bound for AP in the case of
noisy data. And finally it remains an open problem to prove global convergence of our approach (AP + the null initialization). 

These questions are particularly enticing in view of superior numerical performances that strongly indicate 
a large basin of attraction, a high degree of noise tolerance and global
convergence from randomly selected initial data.\\
}


\appendix
 \section{Proof of Theorem ~\ref{Gaussian}}
  The proof is based on the following two propositions.
 \begin{prop} \label{prop5.1}  
 There exists $x_\bot \in \IC^n$ with $x_\bot^* x_0=0$ and $\|x_\bot\|=\|x_0\|=1$  such that 
  \begin{eqnarray}\label{closeness1}
    \frac{1}{4}\|x_0x_0^* -x_{\rm null} x_{\rm null}^* \|^2 &\leq &  {\|b_I\|^2\over \|A_I^* x_\bot \|^2}. \end{eqnarray}

\commentout{
In terms of relative errors,
\[
\min_{\alpha \in \IC,\; |\alpha|=1 }\| x_{\rm null}-\alpha x_0\|^2=2(1-\beta)\le 4 (\frac{1-\beta^2}{2-\beta^2})\le 4\frac{\|b_I\|^2}{\|A^*_I x_\bot\|^2}.
\]
}
 \end{prop}
 \begin{proof}
Since $\xnul$ is optimally phase-adjusted, we have
\beq
\label{52}
\beta: =x_0^*\xnul\ge 0
\eeq
and \beq
\label{a.2}
 x_0=\beta x_{\rm null}+\sqrt{1-\beta^2}\, z 
 \eeq
for some unit vector $z^*x_{\rm null}=0$.
Then \beq
\label{571} x_\bot :=-(1-\beta^2)^{1/2} x_{\rm null}+\beta z 
\eeq
 is a unit vector  satisfying $x_0^*x_\perp=0$. Since $\xnul$ is a singular vector and $z$ belongs in
 another  singular subspace, we have 
  \beqn
\|A_I^* x_0\|^2&=&\beta^2\|A_I^* x_{\rm null}\|^2+(1-\beta^2)\|A_I^* z\|^2, \\
 \|A_I^* x_\bot \|^2&=&(1-\beta^2)\|A_I^*x_{\rm null}\|^2+\beta^2\|A_I^* z\|^2 
  \eeqn
from which it follows that 
  \beq\label{57}
&&(2-\beta^2)\|A_I^* x_0\|^2-(1-\beta^2)\|A_I^* x_\bot \|^2\\
 &=&\|A_I^* x_{\rm null}\|^2+2(1-\beta^2)^2\left(
 \|A_I^* z\|^2-\|A_I^* x_{\rm null}\|^2\right)\ge 0.\nn
 \eeq
 \commentout{
and 
\beq
\label{56'}
\|A_I^* x_\bot \|^2 - \|A_I^* x_0\|^2=
(2\beta^2-1)  \left(
 \|A_I^* z\|^2-\|A_I^* x_{\rm null}\|^2\right)\geq 0
 \eeq
 by \eqref{52}. 
 }
By  \eqref{57}, \eqref{52''} and $\|b_I\|= \|A_I^*x_0\|$, we also have   \begin{eqnarray}
 &&  \frac{\|b_I\|^2}{  \|A_I^* x_\bot \|^2} 
 \ge {1-\beta^2\over 2-\beta^2} \ge \half(1-\beta^2)=  \frac{1}{4}\|x_0x_0^* -x_{\rm null} x_{\rm null}^* \|^2.
 \eeq

 \end{proof}

   \begin{prop}\label{BnormA}
Let  $A\in \IC^{n\times N}$ be  an i.i.d.  complex standard Gaussian random matrix.
Then for any $\ep>0, \delta>0, t>0$
 \[ \|b_I\|^2\le |I| \lt( \left(\frac{2+t}{1-\epsilon} \right) \frac{|I|}{N}+\ep \lt(-2\ln \lt(1-{|I|\over N}\rt)+\delta\rt)\rt)
\]
with probability at least 
\beq
\label{prob'}
&& 1 -2\exp\left(-{N}{\delta^2 e^{-\delta}|1-\sigma|^2/2} \right)-
2\exp\lt(-{2\ep^2 |1-\sigma|^2} \sigma^2 N\rt)-Q \nn
 \eeq
 where $Q$ has the asymptotic upper bound
 \beqn
2 \exp\lt\{-c\min \lt[{e^2t^2 \over 16} {|I|^2\over N} \lt(\ln \sigma^{-1}\rt)^2,~{et\over 4}|I|\ln\sigma^{-1}\rt]\rt\},\quad  \sigma:={|I|\over N}\ll 1. \eeqn\end{prop}
 The proof of Proposition \ref{BnormA} is given in the next section.  \\

Now we turn to the proof of Theorem  \ref{Gaussian}.

Without loss of the generality we may assume $\|x_0\|=1$. Otherwise, we replace $x_0, x_{null}$ by $x_0/\|x_0\|$ and $x_{\rm null}/\|x_0\|$, respectively. 
By an additional orthogonal  transformation which does not affect the statistical nature of the complex Gaussian matrix, we can map $x_0$ to $e_1$, the canonical  vector with 1 as the first entry and zero elsewhere.  \commentout{In this representation, it follows from \eqref{52}-\eqref{571} that 
\[
\beta=x_{\rm null}(1),\quad z(1)=\sqrt{1-\beta^2},\quad x_\perp(1)=0. 
\]
 }

Let $x_\bot$ be any unit vector of the form $
x_\bot =[0,y^\top ]^\top$ where $  y\in \IC^{n-1}$ is an unit vector. 
 Let  $A'\in \IC^{N\times (n-1)}$ be the sub-column matrix of $A_I^*$ with its first column vector deleted and
$\{\nu_{i}\}_{i=1}^{n-1}$ the singular values of $A'$ in the ascending order.
 Let 
 \[
 B'=A'\,\diag(y/|y|)
 \]  which has 
 the same singular values as $A'$.  We have
\[
\|A_I^* x_\bot \|=\|B'\, | y| \|
\] 
and hence
\[
\|A_I^* x_\bot \|=( \| \Re(B') \, |y|  \|^2+\| \Im(B')\,  |y|  \|^2)^{1/2}\ge \sqrt{2}\lt(\| \Re(B') \, |y|  \|\wedge \| \Im(B')\,  |y|  \|\rt).
\]

Note that $A'$ and $B'$ are both i.i.d. complex  Gaussian random matrices for any fixed $y\in \IC^{n-1}$. 
By the  theory of Wishart matrices \cite{DS},  the singular values $\{\nu_j^R \}_{j=1}^{n-1}, \{\nu_j^I \}_{j=1}^{n-1}$ (in the ascending order)  of  $\Re(B'), \Im(B')$ satisfy  the probability bounds
that for every $t>0$ and $j=1,\cdots,n-1$
 \beq
\label{w1} \mathbb{P}\lt(\sqrt{|I|}-(1+t)\sqrt{n}\le \nu_j^R\le \sqrt{|I|}+(1+t)\sqrt{n}\rt)&\ge &1-2e^{-nt^2/2},\\
 \label{w2} \mathbb{P}\lt(\sqrt{|I|}-(1+t)\sqrt{n}\le \nu_j^I\le \sqrt{|I|}+(1+t)\sqrt{n}\rt)&\ge& 1-2e^{-nt^2/2}.
 \eeq
 By Proposition \ref{prop5.1} and \eqref{w1}-\eqref{w2}, we have
 \begin{eqnarray*}
\|x_0x_0^* -x_{\rm null}x_{\rm null}^*\|&\le &{ \sqrt{2}\|b_I\|\over \| \Re(B') \, |y|  \|\wedge \| \Im(B')\,  |y|  \|}\\
&\le&\sqrt{2} \|b_I\|(\nu^R_{n-1}\wedge \nu^I_{n-1})^{-1}\\
&\le& \sqrt{2} \|b_I\|(\sqrt{|I|}-(1+t)\sqrt{n})^{-1}. 
\eeqn
By Proposition \ref{BnormA}, we obtain the desired bound \eqref{error}. 
The success probability is at least the expression \eqref{prob'} minus $4e^{-nt^2/2}$ which equals the expression  
\eqref{prob}.
 

\subsection{Proof of Proposition~\ref{BnormA} }

By the Gaussian assumption, 
   $b(i)^2=|a_i^* x_0|^2$ has a  chi-squared distribution with  the probability density $e^{-z/2}/2$ on $z\in [0,\infty)$ and the cumulative distribution  \begin{eqnarray*}
 F(\tau):=\int_{0}^\tau  2^{-1}\exp(-z/2) dz=1-\exp(-\tau/2). \end{eqnarray*}
 Let  \beq\label{62'} 
\tau_*=-2\ln (1-|I|/N)
\eeq
for which $ F(\tau_*)=|I|/N.$
 
Define
\[ \hat I:=\{i:   b(i)^2\le  \tau_*  \}=\{ i: F(b^2(i))\le |I|/N \},\] and 
\[
\|\hat b\|^2:=\sum_{i\in \hat I} b(i)^2.
\]

Let  \[ \{\tau_1\le \tau_2\le \ldots\le \tau_N\}\] be the sorted sequence of  $\{b(1)^2,\ldots, b(N)^2\}$ in magnitude. 
  \begin{prop}\label{NN} 
{\bf (i)} For any $\delta>0$, we have
   \beq
 \tau_{|I|}&\le& \tau_*+\delta
 \eeq
 with probability at least 
 \beq\label{A3}
 1-\exp\left(-\frac{N}{2} {\delta^2 e^{-\delta} |1-{{|I|/N}}|^2} \right)
   \eeq

{\bf (ii)} For each $\epsilon>0$,
we have
  \beq\label{ZZ} |\hat I| \ge |I| (1-\epsilon)\eeq
  or 
 equivalently,   \beq \label{69} \tau_{\lfloor |I| (1-\epsilon)\rfloor }\le \tau_*
  \eeq
   with probability at least
      \beq
   \label{69.5}
 1-  2\exp\lt(-{4\ep^2 |1-{|I|/N}|^2} |I|^2/N\rt) 
\eeq
   \end{prop}
  \begin{proof} 
  
  {\bf (i)}
 Since  $F'(\tau)= \exp(-\tau/2)/2$, 
 \beq |F(\tau+\epsilon)-F(\tau)|\ge \epsilon/2 \exp(-(\tau+\epsilon)/2).\label{a.16}\eeq
 For $\delta>0$, let
 \[\zeta:=  F(\tau_*+\delta)-F(\tau_*)\]
 which by \eqref{a.16} satisfies
 \beq
 \label{70}
\zeta \ge \frac{\delta}{2}  \exp(-\frac{1}{2}( \tau_*+\delta)). 
\eeq

Let   $\{w_i: i=1,\ldots, N\}$ be  the i.i.d. indicator random variables
   \[
   w_i=\chi_{\{b(i)^2>\tau_*+\delta\}}
   \]
   whose expectation is given by
   \[ \IE[w_i]=1-F(\tau_*+\delta).
   \] 
    The  Hoeffding inequality  yields
     \begin{eqnarray}
\label{ineq1}
 \mathbb{P}(\tau_{|I|} >\tau_*+\delta)
&=&\mathbb{P}\left(\sum_{i=1}^N w_i >N-|I|\right)\\
&=&\mathbb{P}\left(N^{-1}\sum_{i=1}^N w_i-\IE[w_i] >1-|I|/N-\IE[w_i]\right)\nn\\
&=&\mathbb{P}\left(N^{-1}\sum_{i=1}^N w_i-\IE[w_i] >\zeta\right)\nn\\
&\le& \exp(-2N\zeta^2).\nn\end{eqnarray}
Hence, for any fixed $\delta>0$, 
 \beq
  \tau_{|I|}\le& \tau_*+\delta
 \eeq
 holds
 with probability at least 
  \beqn
  \label{a.19}
 1-\exp(-2N\zeta^2)
 &\ge &1-\exp\left(-\frac{N\delta^2}{2} e^{-\tau_*-\delta}\right)\\
  &= &1-\exp\left(-\frac{N \delta^2}{2} e^{-\delta} \left |1-{|I|/N} \right|^2 \right) \nn
   \eeqn
 by \eqref{70}.

{\bf (ii)}    Consider the following replacement
\[
\begin{array}{cl}
(a)   & |I| \longrightarrow \lceil |I| (1-\epsilon)\rceil    \\
(b)   & \tau_* \longrightarrow F^{-1}(\lceil |I| (1-\epsilon)\rceil /N)   \\
(c)   & \delta \longrightarrow F^{-1}(|I|/N)-F^{-1}(\lceil |I| (1-\epsilon)\rceil /N)   \\
(d)   & \zeta \longrightarrow F^{-1}(\tau_*+\delta)-F^{-1}(\tau_*)=|I|/N-\lceil |I| (1-\epsilon)\rceil /N= \frac{\lfloor |I|\epsilon\rfloor}{N}\end{array}
\]
in the preceding argument.
Then (\ref{ineq1}) becomes
 \beqn
\IP\lt(\tau_{\lceil |I|(1-\epsilon)\rceil}>F^{-1}(|I|/N) \rt)& \le & \exp(-2N\zeta^2)
 = \exp\left(-\frac{ 2\lfloor |I|\epsilon\rfloor^2}{N} \right).
 \eeqn
 That is, \[
\tau_{\lceil |I|(1-\epsilon)\rceil}\le \tau_* \]
holds with probability at least
 \beqn
1-\exp(-2 {\lfloor  |I|\epsilon \rfloor^2 }/{N}).
 \eeqn

\end{proof}

   \begin{prop}\label{bhatb}
For each $\epsilon>0$ and $\delta>0$,
 \beq\label{66}
\frac{\|b_I\|^2}{|I|}\le \frac{ \|\hat b\|^2}{|\hat I|}+
\epsilon  (\tau_*+\delta)\eeq
with  probability at least 
\beq \label{77}
 1-2\exp\left(-\frac{1}{2} {\delta^2 e^{-\delta} |1-{|I|/N}|^2 N} \right)-
 2\exp\lt(-{2\ep^2 |1-{|I|/N}|^2}{|I|^2\over N}\rt). 
 \eeq
 
 \end{prop}
 \begin{proof}
Since $\{\tau_j\}$ is an increasing sequence, the function  $T(m)=m^{-1}\sum_{i=1}^m\tau_i$  is also increasing. 
Consider the two alternatives
either   $|I|\ge |\hat I|$ or $|\hat I|\ge |I| $. 
For the latter,   
\[
{\|b_I\|^2}/{|I|}\le { \|\hat b\|^2}/{|\hat I|} 
\]
 due to the monotonicity of $T$.

 For the former case  $|I|\ge |\hat I|$,  we have
 \beqn
T(|I|)&= &|I|^{-1}\sum_{i=1}^{|\hat I|} \tau_i+|I|^{-1}\sum_{i=|\hat I|+1}^{|I|} \tau_i\\
 &\le &T(|\hat I|)+|I|^{-1} (|I|-|\hat I|) \tau_{|I|}.
 \eeqn
By  Proposition \ref{NN}  (ii)  $|\hat I|\geq (1-\ep) |I|$ and hence
 \beqn
T(|I|)&\le &T(|\hat I|)+|I|^{-1}(|I|-|I|(1-\epsilon))\tau_{|I|}
=T(|\hat I|)+\epsilon\tau_{|I|}
\eeqn
 with probability at least given by \eqref{69.5}.

By Proposition \ref{NN} (i),  
$
\tau_{|I|}\le \tau_* +\delta
$
with probability at least given by \eqref{A3}.
\end{proof}

Continuing the proof of Proposition \ref{BnormA},  let us
 consider the i.i.d. centered, bounded random variables \beq
 \label{78}
 Z_i := 
{N^2\over |I|^2} \lt[b(i)^2\chi_{\tau_*}-\IE[b(i)^2\chi_{\tau_*}]\rt] 
 \eeq
 where $\chi_{\tau_*}$ is the characteristic function of the set $\{b(i)^2\leq \tau_*\}$. 
Note that  \beq
\IE(b(j)^2\chi_{\tau_*})&= &\int_{0 }^{\tau_*}  2^{-1} z\exp(-z/2) dz=2-(\tau_*+2)\exp(-\tau_*/2)\le 2 |I|^2/N^2\label{sigmaB} \eeq
and
hence
\beq\label{a.23}
-2 \le Z_i \le \sup\lt\{{N^2\over |I|^2} b(i)^2\chi_{\tau_*} \rt\}= {N^2\over |I|^2} \tau_*.
\eeq

 \commentout{
By the Hoeffding inequality, we then have
  \beq
   \label{a.23'}
   \mathbb{P}\{|\sum_{i=1}^N Z_i|\ge  N t \}\le 2 \exp\lt(-{2N^2 t^2\over N (2+\tau_* N^2/|I|^2)^2}\rt)
   \eeq
}

Next recall the Bernstein-inequality.
   \begin{prop}\cite{Roman} Let $Z_1,\ldots,Z_N$ be i.i.d.  centered sub-exponential
    random variables. Then for every $t\ge 0$, we have
   \beq
   \label{a.21}
   \mathbb{P}\lt\{ N^{-1}|\sum_{i=1}^N Z_i|\ge  t \rt\}\le 2 \exp\lt\{-c \min(Nt^2/K^2, Nt/K)\rt\},
   \eeq
   where $c$ is an absolute constant and 
   \[
   K=\sup_{p\ge 1} p^{-1} (\IE|Z_j|^p)^{1/p}.
   \]
        \end{prop}
        
        \begin{rmk}
   For $K$ we have the following estimates
   \beq\label{a.22}
K&\le& {2N^2\over |I|^2} \sup_{p\ge 1} p^{-1} (\IE| b(i)^2\chi_{\tau_*} |^p)^{1/p} \\
& \le & {2N^2\over |I|^2}  
\tau_* \sup_{p\ge 1} p^{-1}( \IE\chi_{\tau_*})^{1/p}\nn\\
&\le& {2N^2\over |I|^2}  
\tau_* \sup_{p\ge 1} p^{-1}(1-e^{-\tau_*/2})^{1/p}.\nn
\eeq
The maximum of the right hand side of \eqref{a.22} occurs at
\[
p_*=-\ln (1-e^{-\tau_*/2})
\]
and hence
\beqn
K&\le &{2N^2\over |I|^2} {\tau_* \over p_*} (1-e^{-\tau_*/2})^{1/p_*}.
\eeqn
We are interested in the regime
\[
\tau_*\asymp 2|I|/N \ll 1
\]
which implies 
\[
p_*\asymp -\ln {\tau_*\over 2}\asymp \ln{N\over |I|}
\]
and consequently 
\beq
\label{a.24}
K\le {4N\over e |I|} \lt(\ln {N\over |I|}\rt)^{-1},\quad \sigma=|I|/N\ll 1.
\eeq

On the other hand, upon substituting the asymptotic bound \eqref{a.24} in  the probability bound \[
Q=2 \exp\lt\{-c \min(Nt^2/K^2, Nt/K)\rt\}
\]
of \eqref{a.21}, we have 
\[
 K\le 2 \exp\lt\{-c\min \lt[{e^2t^2 \over 16} \lt(\ln \sigma^{-1}\rt)^2 {|I|^2/N},~{et\over 4}|I|\ln\sigma^{-1}\rt]\rt\}, \quad\sigma \ll 1.\\
 \]
     \end{rmk}
        
      The Bernstein  inequality
     ensures  that
  with  high probability
 \[
\left |{ \|\hat b\|^2\over N}- \IE(b^2(i)\chi_{\tau_*})\right|\le t{|I|^2\over N^2}.
 \]
By (\ref{ZZ}) and \eqref{sigmaB}, we also have
   \beq\label{80}
 \frac{\|\hat b\|^2}{|\hat I|}&\le  & \IE(b(i)^2\chi_{\tau_*}) \frac{N}{ |\hat I|}+t\frac{|I|^2}{|\hat I| N}\\
& \le&
\left(\IE(b(i)^2\chi_{\tau_*}) \frac{N^2}{|I|^2}+t\right) \frac{|I|}{N}\nn\\
&\le& \frac{2+t}{1-\epsilon} \cdot\frac{|I|}{N}\nn 
 \eeq

By Prop.~\ref{bhatb}, we now have
\beqn
\|b_I\|^2&\le &|I| \lt({\|\hat b\|^2\over |\hat I|} +\ep \lt(\tau_*+\delta\rt)\rt)
 \eeqn
with probability at least given
by \eqref{prob},
 which together with \eqref{80} and \eqref{62'} complete the proof of Proposition~\ref{BnormA}. \\
 

{\bf Acknowledgements.} 
We thank anonymous referees for helpful suggestions that lead to improvement of
the original manuscript. 


\begin{thebibliography}{99}



\bibitem{Balan2}

R. Balan, B. G. Bodmann ,  P. G. Casazza and  D. Edidin,
``Painless reconstruction from magnitudes of frame coefficients,"
{\em J Fourier Anal Appl} {\bf 15},  488-501 (2009).


\bibitem{Balan1}
R. Balan, P. Casazza and D. Edidin, ``On signal reconstruction without phase,"
{\em Appl. Comput. Harmon.
Anal.} {\bf 20}, 345-356 (2006).


\bibitem{BY}
R. Balan, Y. Wang, ``Invertibility and robustness of phaseless reconstruction", arXiv preprint, arXiv:1308.4718, 2013.

\bibitem{BM2}
 A. S. Bandeira, J. Cahill, D. G. Mixon, A. A. Nelson, ``
Saving phase: Injectivity and stability for phase retrieval,"
{\em Appl. Comput. Harmon. Anal. } {\bf 37}, 106-125 (2014). 

\bibitem{BM1}
A. S. Bandeira, Y. Chen and D. Mixon,
``Phase retrieval from power spectra of masked signals,"
{\em Inform. Infer.}  (2014) 1-20.

\bibitem{BB}
H.H.  Bauschke and J. Borwein, ``On projection algorithms for solving convex feasibility problems,"  {\em SIAM Review} {\bf 38}, 367-426 (1996).

\bibitem{BCL02} 
H.H. Bauschke, P.L. Combettes and D. R. Luke, \lq\lq Phase retrieval, error reduction algorithm, and Fienup variants: a view from convex optimization,"  {\em J. Opt. Soc. Am. A} {\bf 19},  13341-1345 (2002).  


\bibitem{BCL04}
H. H. Bauschke,
P. L. Combettes,
and
D. R. Luke,
``Finding best approximation pairs relative to two
closed convex sets in Hilbert spaces,"  {\em J. Approx. Th.} {\bf 127}, 178-192 (2004).

\bibitem{Np}
D. P. Bertsekas. {\em Nonlinear programming.}  Athena scientific, 2003.



\bibitem{Bregman}
L.M. Bregman, ``The method of successive projection for finding a common point of convex sets," {\em Soviet Math. Dokl.} {\bf 162}, 688-692 (1965).

\bibitem{truncatedWF}
E.J. Cand\`es and Y. Chen, `` Solving random quadratic systems of equations is nearly as easy as solving
linear systems."  arXiv:1505.05114, 2015.

\bibitem{phaselift0}
E.J. Cand\`es, Y.C.  Eldar, T. Strohmer, and V. Voroninski, ``Phase retrieval via matrix completion,"  {\em SIAM J. Imaging Sci.} {\bf  6}, 199-225 (2013). 

\bibitem{CLS2}
E.J. Candes, X. Li and M. Soltanolkotabi, ``Phase retrieval via Wirtinger flow: theory and algorithms," {\em IEEE Trans  Inform. Th. } {\bf 61}(4), 1985--2007 (2015). 

\bibitem{CLS1}
E.J. Candes, X. Li and M. Soltanolkotabi. ``Phase retrieval from coded diffraction patterns." {\em Appl. Comput. Harmon. Anal.} {\bf 39}, 277-299 (2015).

\bibitem{phaselift1}
E.J. Cand\`es, T. Strohmer, and V. Voroninski, `` Phaselift: exact and stable signal recovery from magnitude measurements via convex programming," {\em Comm. Pure Appl. Math.}
{\bf 66}, 1241-1274 (2012).

 
\bibitem{Papa}
 A. Chai, M. Moscoso, G. Papanicolaou, ``Array imaging using intensity-only measurements," {\em Inverse Problems} {\bf 27} (1) (2011).
 
 \bibitem{Chapman11}
H.N. Chapman {\em et al.}  ``Femtosecond X-ray protein nanocrystallography". {\em Nature}  {\bf 470}, 73-77 (2011).

\bibitem{Chapman14}
H.N. Chapman, C. Caleman and N.  Timneanu,  ``Diffraction before destruction," {\em Phil. Trans. R. Soc. B}  {\bf 369} 20130313 (2014). 

\bibitem{DR-phasing}
P. Chen and  A. Fannjiang, ``Phase retrieval with a single mask 
by the Douglas-Rachford algorithm," {\tt  arXiv:1509.00888}.

\bibitem{CG}
W. Cheney and A. Goldstein, ``Proximity maps for convex sets," {\em Proc. Amer. Math. Soc.} {\bf 10},  448-450 (1959).

\bibitem{Cimmino}
G. Cimmino, ``Calcolo approssimato per le soluzioni dei sistemi di equazioni lineari,"
{\em Ric. Sci. Progr. Tecn. Econom. Naz.} {\bf 16} 326-333 (1938).

\bibitem{Conca}
A. Conca, D. Edidin, M. Hering, and C. Vinzant, ``An algebraic characterization of injectivity in phase retrieval," {\em Appl. Comput. Harmon. Anal.} {\bf 38} 346-356 (2015). 

\bibitem{DS}
K.R. Davidson and S.J. Szarek. ``Local operator theory, random matrices and Banach spaces,"  in {\em Handbook of the geometry of Banach spaces}, Vol. I, pp. 317-366. Amsterdam: North-Holland, 2001.

\bibitem{DH12}
L. Demanet and P. Hand, ``Stable optimizationless recovery from phaseless linear measurements,"
{\em J. Fourier Anal. Appl.} {\bf 20}, 199-221 (2014).



\bibitem{Deutsch}
F. Deutsch. {\em Best Approximation in Inner Product Spaces}, Springer, New York, 2001.

\bibitem{Dobson}

D.C. Dobson, `` Phase reconstruction via nonlinear least-squares," {\em Inverse Problems}{\bf 8} (1992) 541-557.

\bibitem{Eldar} Y.C. Eldar and  S. Mendelson, ``Phase retrieval: Stability and recovery guarantees,"  {\em Appl. Comput. Harmon. Anal.} {\bf 36}, pp. 473-494 (2014). 




\bibitem{unique} 
A. Fannjiang, \lq\lq Absolute uniqueness of phase retrieval with random illumination," {\em Inverse Problems} {\bf 28}, 075008 (2012). 

\bibitem{rpi} 
A. Fannjiang and W. Liao, \lq\lq Phase retrieval with random phase illumination," {\em J. Opt. Soc. A} {\bf 29},   1847-1859 (2012).  

\bibitem{pum}
A. Fannjiang and W. Liao, ``Fourier phasing with phase-uncertain mask," {\em Inverse Problems}
{\bf 29} 125001 (2013).

\bibitem{Fie82} 
J. R. Fienup, \lq\lq Phase retrieval algorithms: a comparison,"  {\em Appl. Opt.} {\bf 21},   2758-2769 (1982).  

\bibitem{Fie13}
J.R. Fienup, ``Phase retrieval algorithms: a personal tour ",
{\em Appl. Opt.} {\bf 52} 45-56 (2013).

\bibitem{GS72} 
R.W. Gerchberg and W. O. Saxton, \lq\lq A practical algorithm for the determination of the phase from image and diffraction plane pictures,"  {\em  Optik } {\bf 35},  237-246 (1972).  

\bibitem{Goldstein}
A.A. Goldstein. `` Convex programming in Hilbert space." {\em  Bull. Am. Math. Soc.} {\bf 70}, 709-710 (1964).
\bibitem{Gross}
D. Gross, F. Krahmer and R. Kueng, ``A partial derandomization of phaselift using spherical designs", arXiv:1310.2267, 2013.

\bibitem{Hayes}
M. Hayes,``The reconstruction of a multidimensional sequence
from the phase or magnitude of its Fourier transform," {\em IEEE
Trans. Acoust. Speech Signal Process} {\bf 30}, 140-154 (1982).

\bibitem{Hesse}
R. Hesse, D. R. Luke, S. Sabach, and M.K. Tam,
``Proximal heterogeneous block implicit-explicit method and application to blind ptychographic diffraction imaging," {\em SIAM J. Imag. Sci.} {\bf 8} pp. 426-457 (2015). 







\bibitem{Kac}
S. Kaczmarz, ``Angen\"{a}herte Aufl\"osung von Systemen linearer Gleichungen," {\em Bull. Internat. Acad. Pol. Sci. Lett. Ser. A} {\bf 35}, 355-357 (1937). 

\commentout{
\bibitem{Khatri1}
Khatri, C. G. ``Distribution of the largest or the smallest characteristic root under null hypothesis
concerning complex multivariate normal populations." {\em Ann. Math. Stat. } {\bf 35}, 1807-1810 (1964).
\bibitem{Katri2}
 Khatri, C. G. (1969). Non-central distribution of the i-th largest characteristic roots of three matrices concerning
complex multivariate normal populations. {\em Ann.  Inst. Statistical Math.}  21, 23?32.
}

\bibitem{Klib1}M.K. Klibanov, ``On the recovery of a 2-D function from the modulus of its Fourier transform," {\em J. Math. Anal. Appl.} {\bf 323} 818-843 (2006).

\bibitem{Klib2}
M.K. Klibanov, ``Uniqueness of two phaseless non-overdetermined inverse acoustics problems in 3-d,"{\em Appl. Anal.} {\bf 93} 1135-1149 (2013). 

\bibitem{Levitin}
E.S. Levitin and B.T. Poljak, ``Constrained minimization methods."  {\em U.S.S.R. Comput. Math. Math. Phys.} {\bf 6}, 1-50 (1965).

\bibitem{Luke09}
A.S. Lewis ,  D.R. Luke and   J. Malick, ``Local linear convergence for alternating and
averaged nonconvex projections," {\em Found. Comput. Math.} {\bf 9(4)}, 485--513 (2009)


\bibitem{LV}
X. Li, V. Voroninski, ``Sparse signal recovery from quadratic measurements via convex programming," {\em  SIAM J. Math. Anal.}
{\bf 45} (5), 3019-3033 (2013).




\bibitem{Mar07}
S. Marchesini, `` A unified evaluation of iterative projection algorithms for phase retrieval," {\em Rev. Sci. Instr.} {\bf 78}, 011301 (2007). 


\bibitem{Miao00} 
J. Miao, J. Kirz and D. Sayre, \lq\lq The oversampling phasing method,"  {\em  Acta Cryst. D} {\bf 56}, 1312--1315 (2000).  

\bibitem{MSC}
J. Miao, D. Sayre and H.N. Chapman, ``Phase retrieval from the magnitude of the Fourier transforms of nonperiodic objects,"
{\em J. Opt. Soc. Am. A} {\bf 15} 1662-1669 (1998).

\bibitem{Mig11}
A. Migukin, V. Katkovnik and J. Astola, ``Wave field reconstruction from multiple plane
intensity-only data: augmented Lagrangian algorithm,"
{\em J. Opt. Soc. Am. A} {\bf 28}, 993-1002 (2011).


\bibitem{NJS}
P. Netrapalli, P. Jain, S. Sanghavi, ``Phase retrieval using alternating minimization,"  {\tt arXiv:1306.0160v2}, 2015.

\bibitem{Neuman}
J. von Neuman, {\em Functional Operators Vol. II. The Geometry of Orthogonal Spaces.} Annals of Math. Studies 22. Princeton University Press, 1950. Reprint of notes distributed in 1933.
\bibitem{Hajdu0}

R. Neutze, R. Wouts, D. van der Spoel, E. Weckert, J. Hajdu, ``Potential for biomolecular imaging with femtosecond x-ray pulses," {\em  Nature} {\bf 406} 753-757 (2000).

\bibitem{Noll} D. Noll and A. Rondepierre,
``On local convergence of the method of alternating projections,"
{\em Found. Comput. Math.} {\bf 16}, pp 425-455 (2016). 


\bibitem{Oh}
 H. Ohlsson, A.Y. Yang, R. Dong and  S.S. Sastry, ``Compressive phase retrieval from squared output measurements via semidefinite programming,"  arXiv:1111.6323, 2011.

\bibitem{Vetterli}
 J. Ranieri, A. Chebira, Y.M. Lu, M. Vetterli, ``Phase retrieval for sparse signals: uniqueness conditions," arXiv:1308.3058, 2013.
  
  \bibitem{Schwarz}
 H.A.  Schwarz, ``Ueber einen Grenz\~ubergang durch alternirendes Verfahren," {\em Vierteljahrsschrift der Naturforschenden Gessellschaft in Zurich} {\bf 15}, 272-286 (1870).
   
  \bibitem{Sh}
Y.   Shechtman, Y.C. Eldar, O. Cohen, H.N.  Chapman, M. Jianwei and  M.  Segev, 
``Phase retrieval with application to optical imaging: A contemporary overview,"
{\em IEEE Mag. 
Signal Proc.} {\bf 32}(3) (2015), 87 - 109. 
  
\bibitem{Hajdu}
M.M. Seibert {\em et. al} ``Single mimivirus particles intercepted and imaged with an X-ray laser." {\em Nature} {\bf 470}, U78-U86 (2011). 
\bibitem{Roman}
R. Vershynin. ``Introduction to the non-asymptotic analysis of random matrices."  arXiv preprint arXiv:1011.3027.
\bibitem{Mallat}
I. Waldspurger, A. d'Aspremont and S. Mallat,
``Phase recovery, maxCut and complex semidefinite programming,"
	{\tt arXiv:1206.0102}. 
\bibitem{ADM}
Z. Wen, C. Yang, X. Liu and S. Marchesini, ``Alternating direction methods for classical and ptychographic phase retrieval," {\em  Inverse Problems} {\bf 28},  115010 (2012).

\bibitem{Xin}
P. Yin and J. Xin, ``Phaseliftoff: an accurate and stable phase retrieval method
based on difference of trace and Frobenius norms,"
{\em Commun. Math. Sci.} {\bf 13} (2015).

\commentout{
\bibitem{ptycho08}
P. Thibault, M. Dierolf, A. Menzel, O. Bunk, C. David, F. Pfeiffer, ``High-resolution scanning X-ray diffraction microscopy", { Science} {\bf 321},  379-382 (2008).


\bibitem{Rod10}
F. Zhang and J. M. Rodenburg, ``Phase retrieval based on wave-front relay and modulation,"
{\em Phys. Rev.  B} {\bf  82}, 121104(R) (2010). 
}
\end{thebibliography}

\end{document}